\documentclass[11pt]{article}
\usepackage[T1]{fontenc}
\usepackage{amsfonts}
\usepackage{amsmath}
\usepackage{amssymb}
\usepackage{amsthm}
\usepackage{bbm}
\usepackage{bm}
\usepackage{mathrsfs}
\usepackage{verbatim}
\usepackage{setspace}
\usepackage{color}
\usepackage{pdfsync}
\usepackage{enumitem}
\theoremstyle{plain}
\newtheorem{theorem}{Theorem}[section]
\newtheorem{proposition}[theorem]{Proposition}
\newtheorem{lemma}[theorem]{Lemma}

\theoremstyle{definition}
\newtheorem{definition}[theorem]{Definition}
\newtheorem{remark}[theorem]{Remark}

\newtheorem{example}[theorem]{Example}

\newtheorem{assumption}[theorem]{Assumption}

\theoremstyle{remark}

\renewenvironment{thebibliography}[1]{%
\begin{oldthebibliography}{#1}%
\setlength{\baselineskip}{.9em}
\linespread{1}
\small
\setlength{\parskip}{0.3ex}%
\setlength{\itemsep}{.5em}%
}%
{%
\end{oldthebibliography}%
}

\newcommand{\eps}{\varepsilon}

\newcommand{\E}{\mathbb{E}}
\newcommand{\F}{\mathbb{F}}

\newcommand{\N}{\mathbb{N}}

\newcommand{\Q}{\mathbb{Q}}
\newcommand{\R}{\mathbb{R}}

\newcommand{\cB}{\mathcal{B}}

\newcommand{\cF}{\mathcal{F}}

\newcommand{\cH}{\mathcal{H}}

\newcommand{\cK}{\mathcal{K}}

\newcommand{\cS}{\mathcal{S}}

\newcommand{\fP}{\mathfrak{P}}
\newcommand{\fS}{\mathfrak{S}}
\newcommand{\fM}{\mathfrak{M}}

\newcommand{\cl}{\mathrm{cl}\,}

\DeclareMathOperator{\dom}{dom}
\DeclareMathOperator{\proj}{proj}

\DeclareMathOperator{\NA}{NA}

\newcommand{\sint}{\stackrel{\mbox{\tiny$\bullet$}}{}}

\def\hypo{\mathrm{hypo}\,}

\numberwithin{equation}{section}

\usepackage[pdfborder={0 0 0}]{hyperref}
\hypersetup{
  urlcolor = black,
  pdfauthor = {Ariel Neufeld, Mario \v{S}iki\'c},
  pdfkeywords = {Utility maximization; Knightian uncertainty; Transaction Costs
},
  pdftitle = {Robust Utility Maximization with Friction},
  pdfsubject = {Robust Utility Maximization with Friction},
  pdfpagemode = UseNone
}

\begin{document}

\title{\vspace{-6em}
Robust Utility Maximization in Discrete-Time Markets with Friction
\date{\today}
\author{
  Ariel Neufeld%
  \thanks{
  Department of Mathematics, ETH Zurich, \texttt{ariel.neufeld@math.ethz.ch}.
  Financial support by the Swiss National Foundation grant SNF 200021$\_$153555 is gratefully acknowledged.
  }
  \and
  Mario \v{S}iki\'c
  \thanks{
    Center for Finance and Insurance, University of Zurich, \texttt{mario.sikic@bf.uzh.ch}.
  }
 }
}
\maketitle \vspace{-1.2em}

\begin{abstract}
We study a robust stochastic optimization problem in the quasi-sure setting in discrete-time. We show that under a lineality-type condition the problem admits a maximizer. This condition is implied by the no-arbitrage condition in models of financial markets. As a corollary, we obtain existence of an utility maximizer in the frictionless market model, markets with proportional transaction costs and also more general convex costs, like in the case of market impact.
\end{abstract}

\vspace{.9em}

{\small

}

\section{Introduction}\label{sec:intro}
An agent participates in the market by buying and selling options. If we denote the portfolio by $H$, i.e. the holdings of the agent at time $t$ by $H_t$, then after $T$ steps and liquidation of the portfolio, he or she will have amount $V(H)$ on the bank account. Instead of moving arbitrarily in the market, we ask ourselves whether the optimal portfolio process, or strategy, of the trader exists. When talking about optimality, we need to specify a preference relation on the set of final wealths or states of the bank account. The preference we will be considering is the robust utility preference
\begin{equation}
H= (H_0,\ldots,H_{T-1})\mapsto\inf_{P\in\fP} E^P\big[U(V(H_0,\ldots,H_{T-1}))\big],
\end{equation}
where the utility function is given by $U$ and $\fP$ denotes a collection of probability measures. 

The use of robust utility preference is motivated by the fact that the true probability measure might not be known. So, instead of considering the utility maximization with respect to one measure, that we guess to be the correct one, we consider a family of probability measures that are possible. We then talk about model uncertainty: it is not known which of the probability measures is the true one, but hopefully, the class of probability measures we are considering is big enough to contain the true one. The optimal strategy for the robust utility maximization problem is giving `the best possible performance' under `the worst probability measure'.

Instead of specifying the market model and analyzing the robust utility maximization problem, it will be easier to consider a more general model. The advantage of doing this, beside simpler notation, is the fact that not every utility optimization problem is of the type: maximizing the robust utility of terminal wealth. One can, for instance, consider the problem of optimal consumption stream. We will, thus, consider the following robust optimization problem
\begin{equation}\label{eq:robust-optim-Intro}
\sup_{H\in\cH}\,\inf_{P\in\fP} E^P\big[\Psi(H_0,\ldots,H_{T-1})\big].
\end{equation}
The goal of this paper is to show the existence of a maximizer $\widehat H$ for a general class of concave functions $\Psi$. In Section~\ref{sec:examples}, we will provide examples showing that this more general problem indeed contains the robust utility maximization problem. The examples include the liquidation value of the strategy $H$ at maturity $T$ in frictionless market models, but also with proportional transaction costs and even more general costs.

Concerning the set $\fP$, there are two cases that prompt different approaches. If the set $\fP$ is dominated, i.e. there exists a probability measure $P$ such that every measure $Q\in\fP$ is absolutely continuous with respect to $P$, there is a plethora of approaches to obtain existence of optimizers. A simple one, that works under our set of assumptions: start with the maximizing sequence of strategies and pass, using Komlos type lemma (see e.g. \cite[Lemma~9.8.1, p.202]{DelbaenSchachermayer.06}), to a convergent sequence of convex combinations. The limit of this sequence, thus, needs to be the maximizer. One can also apply dynamic programming, like in~\cite{RasonyiStettner.06}, to obtain existence under relaxed set of assumptions, or a duality argument, see~\cite{Owari.15}.

If the set $\fP$ of probability measures on $\Omega$ is not dominated, the set of approaches to establish existence becomes very limited. Many approaches of the dominated case do not transfer over to this one; for instance, there is no analogue of a Komlos-type lemma. So, the simple argument for existence provided above does not work anymore. Also, the duality approach is, to the best of our knowledge, not applied here. What remains is the dynamic programming approach. For that reason, the setup we are working is in needs to be carefully laid.

Robust utility maximization was already considered in the literature. The closest to our work in discrete-time is~\cite{Nutz.13util}. There, it has been shown existence of an optimal strategy in the frictionless market model where the utility function is defined on the positive half-line. This work is essentially the robust analogue of~\cite{RasonyiStettner.06}. For more results concerning the robust utility maximization problem in a nondominated framework, we refer to \cite{Bartl.16,BiaginiPinar.15,CarassusBlanchard.16,DenisKervarec.13,FouquePunWong.16,LinRiedel.14,MatoussiPossamaiZhou.12utility,NeufeldNutz.15,TevzadzeToronjadzeUzunashvili.13}. To the best of our knowledge, robust utility maximization in the nondominated setting for financial markets with friction has not been studied yet.

The main tool for proving existence is dynamic programming. Dynamic programming is an approach that replaces the multi-step decision problem with a series of one-step decision problems. If one can solve, i.e. prove existence of optimizers of the one-step problems, then one gets existence in general, by using those one-step optimizers in sequence. In order to be able to apply the approach, we need to first work in the setup that is suitable for measurable selection. The formulation of a robust market from~\cite{BouchardNutz.13} turns out to be appropriate for our needs. We will follow the formulation of dynamic programming in~\cite{Evstigneev.76}.

How does one establish the existence of the optimal strategy in the one-step models that we need to solve? One approach, taken in~\cite{Nutz.13util} and~\cite{RasonyiStettner.06}, is to set up the problem in such a way that the set of strategies in the one-step problems one obtains is compact. Indeed, under a suitable no-arbitrage condition, having a frictionless market model with utility function defined on the positive half-line implies this compactness, up to the projection on the predictable range.

To consider utility functions that are defined on the whole real line, for example the exponential utility function, one needs to resort to convex analysis. Theorem~9.2 in~\cite{Rockafellar.97} provides the necessary condition. When considering utility maximization in a one-step market model, this condition turns out to be just the no-arbitrage condition; we will prove this statement in the course of the paper. In our general setting, the condition that will serve as a no-arbitrage condition will be the following
\begin{equation}\label{eq:def-NA-Intro}
\mbox{the set}
\quad
\cK:=\big\{H \in \cH \, \big| \, \Psi^\infty(H_0,\dots,H_{T-1})\geq 0 \ \fP\mbox{-q.s.}\big\}
\quad
\mbox{ is linear,}
\end{equation}
where here $\Psi^\infty$ denotes  the horizon function of the concave function $\Psi$; see 
Section~\ref{sec:Multi-Per}. This condition coincides with the robust no-arbitrage condition in the frictionless market model. The main step, indeed the main technical obstacle in the proof of our main result, is to show that the `global' condition~\eqref{eq:def-NA-Intro} satisfy a `local' version at each time step.

The remainder of this paper is organized as follows. In Section 2, we introduce the concepts, list the assumptions imposed on $\Psi$ and state the main results. The examples of robust utility maximization in different financial markets with friction are given in Section~3. In Section~4, we introduce and solve the corresponding one-period maximization problem. The problem of finding a maximizer reduces to a question of closedness property of the hypograph of the function over which one maximizes. The result \cite[Theorem~9.2, p.75]{Rockafellar.97} provides an answer to the above question and explains the sufficiency of the condition \eqref{eq:def-NA-Intro}. In Section~5, we introduce the notion needed in our dynamic programming approach and explain why this leads to the existence of a maximizer in our optimization problem \eqref{eq:robust-optim-Intro}. The proof is then divided into several steps, which heavily uses the theory of lower semianalytic functions.

\paragraph{\textbf{Notation}}

For any vector $x\in\R^{dT}$, written out as $x=(x_0,\ldots,x_{T-1})$, where $x_i\in\R^d$ for each $i$, we denote the restriction to the first $t$ entries by $x^t:=(x_0,...,x_{t-1})$. For $y\in \R^{dT}$, we denote by $x\cdot y$ the usual scalar product on $\R^{dT}$.

\section{Optimization Problem}\label{sec:Multi-Per}
Let $T\in \N$ denote the fixed finite time horizon and let $\Omega_1$ be a Polish space. Denote by $\Omega^t:=\Omega^t_1$ the $t$-fold Cartesian product for $t=0,1,\dots,T$, where we use the convention that $\Omega^0$ is a singleton. Let $\cF_t:=\bigcap_P \cB(\Omega^t)^P$ be the universal completion of the Borel $\sigma$-field $\cB(\Omega^t)$; here $\cB(\Omega^t)^P$ denotes the $P$-completion of $\cB(\Omega^t)$ and $P$ ranges over the set $\fM_1(\Omega^t)$ of all probability measures on $(\Omega^t,\cB(\Omega^t))$. Moreover, define $(\Omega,\cF):=(\Omega^T,\cF_T)$. This plays the role of our initial measurable space.

For every $t \in  \{0,1,\dots,T-1\}$ and $\omega^t \in \Omega^t$ we fix a nonempty 
set $\fP_t(\omega^t)\subseteq\fM_1(\Omega_1)$ of probability measures; $\fP_t(\omega^t)$ represents the possible laws
for the $t$-th period given state $\omega^t$.  Endowing $\fM_1(\Omega_1)$ with the usual topology induced by the weak convergence makes it into a Polish space; see  \cite[Chapter~7]{BertsekasShreve.78}. We assume that for each $t$
\begin{equation*}
\mbox{graph}(\fP_t):=\{(\omega^t,P)\,| \, \omega^t \in \Omega^t,\, P\in \fP_t(\omega^t)\}\ \ \mbox{ is an analytic subset of }\ \ \Omega^t \times \fM_1(\Omega_1) .
\end{equation*}
Recall that a subset of a Polish space is called analytic if it is the image of a Borel subset of a (possibly different) Polish space under a Borel-measurable mapping (see \cite[Chapter~7]{BertsekasShreve.78}); in particular, the above assumption is satisfied if  $\mbox{graph}(\fP_t)$ is Borel. The set $\mbox{graph}(\fP_t)$ being analytic provides the existence of an universally measurable kernel $P_t:\Omega^t\to \fM_1(\Omega_1)$ such that $P_t(\omega^t)\in \fP_t(\omega^t)$ for all $\omega^t \in \Omega^t$ by the Jankov-von Neumann theorem, see \cite[Proposition~7.49, p.182]{BertsekasShreve.78}. Given such a kernel $P_t$ for each $t\in \{0,1,\dots,T-1\}$, we can define a probability measure $P$ on $\Omega$ by 
\begin{equation*}
P(A):=\int_{\Omega_1}\dots \int_{\Omega_1} \mathbf{1}_A(\omega_1,\dots,\omega_T)\, P_{T-1}(\omega_1,\dots,\omega_{T-1};d\omega_T)\dots P_0(d\omega_1), \quad A \in \cF,
\end{equation*}
where we write $\omega:=\omega^T:=(\omega_1,\dots,\omega_T)$ for any element in $\Omega$. We denote a probability measure defined as above by  $P=P_0 \otimes \dots \otimes P_{T-1}$.  For the  multi-period market, we consider the set 
\begin{equation*}
\fP:=\{P_0\otimes\dots\otimes P_{T-1}\,|\,  P_t(\cdot) \in \fP_t(\cdot), \, t=0,\dots,T-1\}\subseteq \fM_1(\Omega),
\end{equation*}
of probability measures representing the uncertainty of the law,
where in the above definition each $P_t:\Omega^t\to \fM_1(\Omega_1)$ is universally measurable such that $P_t(\omega^t)\in \fP_t(\omega^t)$ for all $\omega^t \in \Omega^t$.

We will often interpret $(\Omega^t,\cF_t)$ as a subspace of $(\Omega,\cF)$ in the following way. 
Any set $A \subset \Omega^t$ can be extended to a subset of $\Omega^T$ by adding $(T-t)$ products of $\Omega_1$, i.e. $A^T:=A \times \Omega_1\times\dots \times\Omega_1\subset \Omega^T$. Then, for every measure $P=P_0\otimes\dots\otimes P_{T-1} \in \fP$,
one can associate a measure $P^t$ on $(\Omega^t,\cF^t)$ such that $P^t[A]=P[A^T]$ 
by setting $P^t:=P_0\otimes \dots \otimes P_{t-1}$. 

We call a set $A\subseteq \Omega\,$ $\fP$-polar if $A\subseteq A'$ for some $A'\in \cF$ such that $P[A']=0$ for all $P \in \fP$, and say a property to hold $\fP$-quasi surely, or simply $\fP$-q.s., if the property holds outside a $\fP$-polar set. 

A map $\Psi\colon\Omega\times\R^{dT}\rightarrow\overline{\R}$ is called an $\cF$-\emph{measurable normal integrand} if the measurable correspondence $\hypo \Psi:\Omega\rightrightarrows\R^{dT}\times\R$ defined by
$$
\hypo \Psi(\omega)=\big\{(x,y)\in\R^{dT}\times\R\,\big|\,\Psi(\omega,x)\geq y\big\}
$$ 
is closed valued and $\cF$-measurable in the sense of set-valued maps, see \cite[Definition~14.1 and Definition~14.27]{RockafellarWets}. Note that the correspondence $\hypo \Psi$ has closed values if and only if the function $x\mapsto \Psi(\omega,x)$ is upper-semicontinuous for each $\omega$; see \cite[Theorem~1.6]{RockafellarWets}. By \cite[Corollary~14.34]{RockafellarWets}, 
$\Psi$ is (jointly) measurable with respect to $\cF\otimes\cB(\R^{dT})$ and $\cB(\overline \R)$. Classical examples of normal integrands, which are most prevalent in mathematical finance, are Caratheodory maps; see  \cite[Example~14.29]{RockafellarWets}.

Denote by $\cH$ the set of all $\F$-adapted $\R^d$-valued processes $H:=(H_0,\dots,H_{T-1})$ with discrete-time index $t=0,\dots,T-1$. Our goal is to study the following control problem
\begin{equation}\label{eq:optim-prob-multiP}
\sup_{H \in \cH}\inf_{P \in \fP} E^P[\Psi(H_0,\dots,H_{T-1})],
\end{equation}
where $\Psi\colon\Omega\times \R^{d T} \to \R$ is a concave and $\cF$-measurable normal integrand.

Recall that a function $f$ from a Borel subset of a Polish space into $[-\infty,\infty]$ is called lower semianalytic if the set $\{f<c\}$ is analytic for all $c\in \R$; in particular any Borel function is lower semianalytic. A concave function $f\colon\R^n\to \R \cup\{-\infty\}$ is called proper if $f(x)>-\infty$ for some $x \in \R^n$. We refer to \cite{BertsekasShreve.78} and \cite{Rockafellar.97,RockafellarWets} for more details about the theory of lower semianalytic functions and  convex analysis, respectively.

The following conditions are in force throughout the paper.
\begin{assumption}\label{ass:Psi-MultiP}
The map $\Psi:\Omega\times \R^{d T} \to \R\cup\{-\infty\}$ satisfies the following
\begin{enumerate}
\item[(1)] for every $\omega \in \Omega$, the map $x \mapsto \Psi(\omega,x)$ is concave and upper-semicontinuous;
\item[(2)] there exists a constant $C \in \R$ such that $\Psi(\omega,x)\leq C$ for all $\omega \in \Omega$, $x \in \R^{d T}$;
\item[(3)] the map $(\omega,x) \mapsto \Psi(\omega,x)$ is lower semianalytic;
\item[(4)] there exists $h^\circ \in \R^{dT}$, an $\varepsilon>0$ and a constant $c>0$ such that 
\end{enumerate}
\begin{align} \label{eq:interior}
\Psi(\omega,x)\geq -c\qquad \forall\omega\in\Omega,\,\forall x\in\R^{dT}:\,\,\|x-h^{\circ}\|\leq\varepsilon.
\end{align}
\end{assumption}
\begin{remark}\label{rem:bounded above}
At first glance, Assumption~\ref{ass:Psi-MultiP}(2) may look to be rather restrictive. However, it was shown in \cite[Example~2.3]{Nutz.13util}  that  for any (nondecreasing, strictly concave) utility function $U$ being \textit{unbounded from above}, one can construct a frictionless market $S$ and a set $\fP$ of probability measures such that
\begin{equation*}
u(x):=\sup_{H \in \cH}\inf_{P \in \fP} E^P[U(x+ H \sint S_T)]<\infty
\end{equation*}
for any initial capital $x>0$,
but there is no maximizer $\widehat{H}$. So already in the special case of 
$\Psi(H):=U(V(H))$ in the frictionless market, the existence may fail for utility functions not being bounded from above. 
\end{remark}
\begin{remark}
\label{rem:ass-normal-multiP}
The mapping $\Psi$ satisfying Assumption~\ref{ass:Psi-MultiP} is a normal integrand. Indeed,  Assumption~\ref{ass:Psi-MultiP}(3) and \cite[Lemma~7.29, p.174]{BertsekasShreve.78} imply that the map $\omega \mapsto \Psi(\omega,x)$ is $\cF$-measurable for every $x\in\R^{dT}$; recall that $\cF$ is the universal $\sigma$-field on $\Omega$. Therefore, normality of $\Psi$ follows directly from \cite[Proposition~14.39, p.666]{RockafellarWets}.
\end{remark}

\begin{remark}
Assumption~\ref{ass:Psi-MultiP}(4) is primarily a statement saying that the domain $\dom\Psi(\omega)$ of the mapping $x\mapsto\Psi(\omega,x)$ has an interior for each $\omega\in\Omega$. It is even stronger, since the strong lower bound on $\Psi$ on the neighborhood of $h^{\circ}$ implies that every strategy $H\in\cH$ satisfying $\|H(\omega)-h^{\circ}\|\leq\varepsilon$ for every $\omega\in\Omega$ satisfies also $\inf_{P\in\fP}\E^P[\Psi(H)]\geq -c$. 
This is a strong regularity condition that plays a crucial role in establishing measurability of the objects arising in the dynamic programming procedure. More precisely, if $h^\circ=0$, it implies that it is enough to know the market on the countable set of deterministic strategies $H\in\Q^{dT}$; c.f. Remark~\ref{ex:constraint}. It can certainly be relaxed, however it is difficult to come up with a set of conditions on $\Psi$ that could be checked a priori. The most general sets of conditions should be stated in terms of objects arising in the dynamic programming procedure; see Remark~\ref{rem:on assumption 2.1}(3) and the proof of Lemma~\ref{le:Psi-t-lsa-MultiP}.
\end{remark}

\begin{remark}\label{rem:normal-differ}
We point out that our definition of a normal integrand $\Psi$ varies from the classical one in convex optimization as defined e.g. in~\cite[Chapter 14]{RockafellarWets}, in the sense that $-\Psi$ is a normal intergrand in  classical convex analysis. As we are looking for a maximum of a concave function, our definition of a normal function fits into our setting.
\end{remark}

We now define the horizon function
$\Psi^\infty:\Omega\times\R^{dT}\rightarrow\R\cup\{-\infty\}$ of a concave, proper, upper-semicontinuous  integrand $\Psi(\omega,\cdot)$ by
\begin{equation*}
\Psi^\infty(\omega,h) 
= \lim_{n\rightarrow\infty}\frac1n[\Psi(\omega,x+nh)-\Psi(\omega,x)]
=\inf_{n\in \N}\frac1n[\Psi(\omega,x+nh)-\Psi(\omega,x)]
\end{equation*}
where $x\in \R^{dT}$ is any vector with $\Psi(\omega,x)>-\infty$.
Note that for any fixed $\omega\in \Omega$, the map $h\mapsto\Psi^\infty(\omega,h)$ does not depend on the choice of $x$ in the definition; see~\cite[Theorem~3.21, p.87]{RockafellarWets}. The mapping $\Psi^\infty(\omega,\cdot)$ is positively homogeneous, concave and upper-semicontinuous, see~\cite[Theorem~3.21, p.87]{RockafellarWets}. If in addition, $\Psi$ is normal, then so is $\Psi^\infty$, see \cite[Exercise 14.54(a), p.673]{RockafellarWets}.

Throughout the paper we impose the following  condition.
\begin{equation}\label{eq:no-arbitrage-multiP}
\mbox{the set}
\quad
\cK:=\{H \in \cH \, | \, \Psi^\infty(H_0,\dots,H_{T-1})\geq 0 \ \fP\mbox{-q.s.}\}
\quad
\mbox{ is linear.}
\tag{\ensuremath{\NA(\fP)}}
\end{equation}
We call it the no-arbitrage condition.
Of course, the set $\cH$ of adapted strategies is a linear space, hence $\cK$ is a subset of it and the definition makes sense. Note that the set $\cK$ is a convex cone; this follows directly from concavity of the map $\Psi^\infty$. 
\begin{remark}\label{rem:NA}
Naming the condition  $\NA(\fP)$ a (robust) no-arbitrage condition is motivated by the following observation: consider a frictionless market with corresponding price process $S$
and let
\begin{equation*}
\Psi(H):= \sum_{t=0}^{T-1} H_t \cdot(S_{t+1}-S_t)=: H \sint S_T
\end{equation*}
denote the capital gains from trading in the market using strategy $H \in \cH$, our notion of $\NA(\fP)$ coincides with the robust no-arbitrage notion in \cite{BouchardNutz.13,Nutz.13util}. This follows directly since linearity of $\Psi$ in this situation implies that $\Psi^\infty= \Psi$.
Indeed, if $\cK$ were not linear, there would exist a strategy $H\in\cK$  such that $-H\not\in\cK$. But this implies that $H \sint S_T\geq 0 \ \fP$-q.s. and there exists a measure $P\in\fP$ with $P[-H\sint S_T< 0]>0$, i.e.  $P[H\sint S_T>0]>0$, which means that $H$ is a  robust arbitrage strategy in the sense of \cite{BouchardNutz.13,Nutz.13util}; and vice versa.

Moreover, when the map $\Psi$ is of the form initially considered, i.e. given by
\begin{equation*}
\Psi(H)=U(V(H)),
\end{equation*}
a sufficient condition (independent of the utility function $U$!) for $\Psi$ to satisfy  $\NA(\fP)$  is  $V$ satisfying the following condition
\begin{equation*}
\mbox{the set}
\quad
\{H \in \cH \, | \, V^\infty(H)\geq 0 \ \fP\mbox{-q.s.}\}
\quad
\mbox{ is linear.}
\end{equation*}
This is a condition on the financial market model, where  $V(H)$ denotes the terminal wealth when investing with strategy $H$. Indeed, let $U$ be a nondecreasing, concave utility function such that 
$\Psi$ is not identically equal to $-\infty$. Then by Lemma~\ref{le:comp-horizon-append} 
\begin{equation*}
\Psi^\infty(h)=\begin{cases}
U^\infty(V^\infty(h)) & \mbox{if } \  V^\infty(h)>-\infty\\
-\infty& \mbox{otherwise}.
\end{cases}
\end{equation*}
We claim that linearity of the set $\{H \in \cH \, | \, V^\infty(H)\geq 0 \ \fP\mbox{-q.s.}\}$ implies linearity of the set $\{H \in \cH \, | \, \Psi^\infty(H)\geq 0 \ \fP\mbox{-q.s.}\}$. To see this, observe that $U^\infty(0)=0$, and $U^\infty$ is nondecreasing as $U$ is so, too. Let $H \in \cH$ satisfy $\Psi^\infty(H) \geq 0 \ \fP$-q.s. By the monotonicity of $U^\infty$ and as $\Psi^\infty=U^\infty \circ V^\infty$, this means that $V^\infty(H)\geq 0  \ \fP$-q.s. By assumption on the linearity of the set $\{H \in \cH \, | \, V^\infty(H)\geq 0 \ \fP\mbox{-q.s.}\}$ we have $V^\infty(-H)\geq 0  \ \fP$-q.s., which implies that $U^\infty(V^\infty(-H)) \geq 0  \ \fP$-q.s.. Linearity of the set $\{H \in \cH \, | \, \Psi^\infty(H)\geq 0 \ \fP\mbox{-q.s.}\}$ now follows.
%
%
%

Generalized notions of no-arbitrage conditions in form of linearity type conditions were already obtained in~\cite{MR2828763,PennanenPerkkio.12} 
for markets without uncertainty (i.e. where one fixed measure $P$ is given). Our no-arbitrage condition $\NA(\fP)$ can be interpreted as an extension of linearity-type of no-arbitrage conditions to the robust framework.
\end{remark}

The main theorem of this paper is the following.
\begin{theorem}\label{thm:Maxim-exist-MultiP}
Let $\Psi$ be a map satisfying Assumption~\ref{ass:Psi-MultiP}. If the no-arbitrage condition $NA(\fP)$ holds, then there exists a process $\widehat{H} \in \cH$ such that
\begin{equation}\label{eq:thm-optimal-MultiP}
\inf_{P \in \fP} E^P[\Psi(\widehat{H}_0,\dots,\widehat{H}_{T-1})]= \sup_{H \in \cH}\inf_{P \in \fP} E^P[\Psi(H_0,\dots,H_{T-1})].
\end{equation}
\end{theorem}
We will give the proof of this theorem in Section~\ref{sec:proof-multi-period}.
%
%
\section{Examples}\label{sec:examples}
In this section, we give several examples of robust utility maximization in various models of financial markets fitting into the setting of Theorem~\ref{thm:Maxim-exist-MultiP}. This was our initial motivation for the abstract robust optimization problem. 
\begin{example}\label{ex:frictionless}
In this example, we analyze the robust utility maximization problem in a classical frictionless market similar to \cite{Nutz.13util}. 

Let $S=(S^1,\dots,S^d)$ be a $d$-dimensional stock price process with nonnegative components being Borel-measurable and constant $S^j_0=s^j_0>0$ for all $j$.
Consider a random utility function $U\colon\Omega\times\R\rightarrow\R \cup\{-\infty\}$, i.e. $U(\omega,\cdot)$ is a nondecreasing, concave function, which is upper-semicontinuous and bounded from above by a constant. Moreover, assume that $(\omega,y)\mapsto U(\omega,y)$ is lower semianalytic and $\omega\mapsto U(\omega,y)$ is bounded from below for each $y>0$; the last two conditions are trivially satisfied if $U\colon\R \rightarrow \R\cup\{-\infty\}$ is a classical utility function independent of $\omega$ satisfying $U(y)>-\infty$ for $y\in(0,\infty)$. We define the mapping $\Psi$ by
$$
\Psi(H) = U(x + H\sint S_T),
$$
where $x>0$ is the fixed initial wealth of the trader.
We want to show that $\Psi$ we just defined satisfies Assumption~\ref{ass:Psi-MultiP}.  To that end, note that it is concave as $\Psi(\omega,\cdot)$ is a compositum of a concave and a linear function. Also upper-semicontinuity is clear as $U$ is 
upper-semicontinuous and $H\mapsto H \sint S_T$ is continuous for every $\omega$. 
As the utility function is bounded from above, the same holds for the mapping $\Psi$. Due to the assumption on $U$ being lower semianalytic, the same holds true for $\Psi$ being a precomposition of a lower semianalytic function with a Borel function; see \cite[Lemma~7.30(3), p.177]{BertsekasShreve.78}. 

To see that Assumption~\ref{ass:Psi-MultiP}(4) is satisfied, set
$$
\rho := \frac{x}{2dT\max\{ s^j_0 \,:\,j=1,\ldots,d\}}
$$
and define the deterministic strategy $h^{\circ}_t:=(h^{\circ,1}_t,\dots,h^{\circ,d}_t) \in \R^{dT}$ by  
%
$$
h^{\circ,j}_t := (T-t)\rho,\qquad \mbox{ for }t=0,\ldots,T-1,\ \ j=1,\ldots,d. 
$$
Let $\varepsilon<\frac\rho3$. It is easy to see that any $z\in\R^{dT}$ satisfying $\|z-h^\circ\|\leq\varepsilon$ is decreasing in each of the components; i.e. the (deterministic) process  $(z^j_t)_{t\in \{0,\dots,T-1\}}$ is decreasing for each $j=1,\ldots,d$. We claim that the corresponding capital gains $z\sint S_T$ at time $T$ are uniformly bounded from below. To see this, let $z\in\cH$ be one of such (deterministic) strategies. Writing $z_{-1}=0=z_T$, we obtain the corresponding capital gains
$$
z\sint S_T 
=
\sum_{t=0}^{T-1} z_t\cdot(S_{t+1}-S_t)
= 
- \sum_{t=0}^{T} (z_t-z_{t-1}) \cdot S_t.
$$
By assumption on  nonnegativity of each stock price process $S^j$ and the fact that the strategy $(z_t^j)_{t\in \{0,\dots,T-1\}}$ is decreasing, it follows that $-\sum_{t=1}^{T} (z_t-z_{t-1}) \cdot S_t\geq 0$. To see that also $-(z_0-z_{-1}) \cdot S_0$ is bounded from below, we use the definition of $\rho$ to see that
$$
-(z_0-z_{-1}) \cdot S_0
=
-\sum_{j=1}^d(z_0^j-z_{-1}^j) S_0^j
=
-\sum_{j=1}^d z_0^j S_0^j
\geq 
-\sum_{j=1}^d\rho\left(\frac13+T\right)s_0^j
> -x.
$$
Thus, the claim holds true, i.e. for some constant $\delta$, $x+z\sint S_T\geq \delta> 0$ for each $(z_t)$ satisfying $\|z-h^\circ\|\leq\varepsilon$, and each $\omega$. Therefore,  $\Psi$ satisfies  Assumption~\ref{ass:Psi-MultiP}(4), as we assumed that $\omega\mapsto U(\omega,\delta)$ is bounded from below.
\end{example}

\begin{remark}
Our main theorem is more general than the result from \cite{Nutz.13util} in a few directions. First, we do not assume that the stock price process is adapted and therefore also include e.g. the setup of
\cite{kabanov2006dalang}. Moreover, we do not impose 
the assumption that the utility function is defined on the positive half-line. This allows one to include, for instance, the exponential utility function. 
\end{remark}

\begin{remark}\label{ex:constraint}
One models portfolio constraints by correspondences $D_t\colon\Omega^t\rightrightarrows\R^d$ requiring that $H_t(\omega)\in D_t(\omega) \ \forall \omega,t$.
Convex constraints $D_t\colon\Omega^t\rightrightarrows\R^d$, which are given by Borel measurable, closed-valued correspondences are included in our model as long as either $D_t=\{h^\circ\}$ or $D_t$ has an interior in such a way that Assumption~\ref{ass:Psi-MultiP}(4) is satisfied. However, one can read out from the proof of our main theorem that Assumption~\ref{ass:Psi-MultiP}(4) can be relaxed for constraints, i.e. it is not necessary for the sets $D_t$ to include a ball around $h^\circ$, it is enough for them to have an open interior for the dynamic programming with the criterion function
$$
\widehat\Psi(\omega,h)
= \Psi(\omega,h) - \sum_{t=0}^{T-1} \chi_{D_t(\omega)}(h_t)
$$
to give the existence of a maximizer; here $\chi_A$ denotes the convex analytic indicator function for the set $A\equiv D_t(\omega)$ giving the value 0 if $h_t\in A$ and $\infty$ otherwise. 
\end{remark}

\begin{example}
In this example, we consider the financial market from 
\cite{DolinskySoner.13} with  proportional transaction costs. The mark-to-market value of the portfolio strategy $H$ with initial (fixed) capital $x> 0$ is defined by
$$
V(H) = x + \sum_{t=0}^{T-1} H_t\cdot(S_{t+1}-S_t) - \kappa S_t |H_t-H_{t-1}|,
$$
where we set $H_{-1}=0$. The one-dimensional stock price process $(S_t)$ is assumed to be Borel-measurable 
and nonnegative starting at a constant $S_0=s_0>0$. The constant $1>\kappa\geq 0$ indicates the amount of transaction costs. One then defines  $\Psi$ to be
$$
\Psi(\omega,H) = U(V(H))
$$
for a random utility function $U\colon\Omega\times\R\rightarrow\R\cup\{-\infty\}$ being defined as in Example~\ref{ex:frictionless}. 
It is easy to check that the conditions of Assumption~\ref{ass:Psi-MultiP} are satisfied. Concavity and upper-semicontinuity is clear by the fact that $U$ is upper-semicontinuous, concave, nondecreasing and $V$ is concave and continuous. Boundedness from above is clear by the same assumption on the utility function $U$. The lower semianalyticity of $\Psi$ is fulfilled as $S$ is assumed to be Borel.  Now, rewrite the value of the strategy as follows
\begin{align*}
V(H)
&= 
x - \sum_{t=0}^{T-1} S_t\big(\kappa  |H_t-H_{t-1}| + (H_t-H_{t-1})\big) + H_{T-1}S_T
\\
&= 
x - \sum_{t=0}^{T-1} S_t f(H_t-H_{t-1}) + H_{T-1}S_T,
\end{align*}
where the function $f\colon h\mapsto\kappa  |h| + h$ is less than or equal to zero for all $h\leq 0$ by assumption on the constant $\kappa$. Define $(h^\circ_t)_{t=0,\dots,T-1} \in \R^T$ and $\rho$ as in Example~\ref{ex:frictionless}, and choose $\varepsilon<\rho\min\{\frac{1}{3},\frac{T(1-\kappa)}{1+\kappa}\}$. We know from Example~\ref{ex:frictionless} that any $H \in \cH$ satisfying  $\|H-h^\circ\|\leq\varepsilon$ is positive and decreasing, hence
\begin{equation*}
-\sum_{t=1}^{T-1} S_t f(H_t-H_{t-1}) + H_{T-1}S_T\geq 0.
\end{equation*}
Moreover, as $\varepsilon<\rho \frac{T(1-\kappa)}{1+\kappa}$, it is straightforward to see that
\begin{equation*}
-S_0 f(H_0-H_{-1})=-s_0 (\kappa H_0 + H_0)\geq -s_0(T\rho+\varepsilon)(1+\kappa)>-x.
\end{equation*}
Therefore, there exists a constant $\delta$ such that $V(H)\geq\delta>0$ for all $\omega$ and all $H$ satisfying $\|H-h^\circ\|\leq\varepsilon$. Hence $\Psi=U\circ V$ satisfies Assumption~\ref{ass:Psi-MultiP}(4), as we assumed that $\omega\mapsto U(\omega,\delta)$ is bounded from below.
\end{example}
%
%
%
\begin{remark}
In Remark~\ref{rem:NA} we motivated the name `no-arbitrage' for the condition $\NA(\fP)$ by indicating that it coincides with the definition of robust no-arbitrage condition introduced in~\cite{BouchardNutz.13}; we sketch a similar argument for the case of proportional transaction costs. Questions of no-arbitrage in models with proportional transaction costs are usually addressed in the setting of~\cite{Kabanov.99}. The fundamental contribution to the no-arbitrage theory in the setup without model uncertainty was made in~\cite{Schachermayer.04} where the concept of robust no-arbitrage was introduced. More precisely, it was shown that its robust no-arbitrage condition is equivalent to linearity of a certain set of portfolio rebalancings; see~\cite[Lemma~2.6]{Schachermayer.04}.
One can translate the example above to the framework of~\cite{Kabanov.99}.
Observe that our condition $\NA(\fP)$ in fact generalizes the robust no-arbitrage condition of~\cite{Schachermayer.04}, i.e. an equivalent formulation thereof, to this non-dominated setup.
\end{remark}
\begin{remark}
The example above treated the mark-to-market value in the market. One could equivalently consider the liquidation value of the portfolio $H$. Also, it is easy to see that this is not restricted to proportional transaction costs, but can be extended to more general transaction costs. The important thing is the condition Assumption~\ref{ass:Psi-MultiP}(4) and the example above indicates where it comes up in the argument.
\end{remark}

\begin{remark}
One could also consider (proportional) transaction costs in physical units.
We will only sketch this example.
There are $d$ risky assets in the market and the portfolio of the trader is described by specifying at each time $t$ the number of shares in each of the $d$ risky assets. Going back to the original contribution of 
\cite{Kabanov.99}, one models trading in the market by specifying in each time instance $t=0,\ldots,T-1$ how many shares $H_t^{\alpha\rightarrow \beta}$ of risky asset $\alpha$ to transfer to shares of the risky asset $\beta$. Hence, a strategy will be a matrix with adapted entries. The market mechanism, i.e. changes of the portfolio due to a trade order, is given by a sequence of maps
$$
F_t:\Omega^t\times\R^{d\times d}\longrightarrow \R^d.
$$
The interpretation is the following: after executing the order $H$ at time $t$, the holdings of the trader in physical units are going to change by $F_t(H).$ So, the portfolio of the trader at the end of the trading period is given by
$$
V(H) = \sum_{t=0}^{T-1} F_t(H_t).
$$
We also refer to \cite{Bouchard.06}. 

Now, let the utility function be denoted $U\colon\R^d\rightarrow\R$, and define the map $\Psi$  by
$$
\Psi(H) = U(V(H)).
$$
One can easily find conditions under which Assumption~\ref{ass:Psi-MultiP} is satisfied. For instance: the utility function $U$ is bounded from above, concave, nondecreasing, and upper-semicontinuous with respect to partial order given by an order cone $\R_+^d$; the market impact functions $F_t$ being continuous and concave with respect to the same partial order; appropriate measurability conditions and boundedness conditions to obtain conditions Assumption~\ref{ass:Psi-MultiP}(3,4).
\end{remark}

\begin{example} 
In this example, we consider the optimal liquidation problem, an adaptation of the model introduced in~\cite{AlmgrenChriss.00}. The agent initially holds $X>0$ units of a security, which need to be liquidated by the time of maturity $T$. The strategy $H=(H_0,\dots,H_{T-1})$  denotes that the agent holds $H_t$ units of the security after liquidation of the amount $H_{t-1}-H_t$ at time $t$; in the interpretation of~\cite{AlmgrenChriss.00} the agent liquidates this amount between times $t-1$ and $t$. This liquidation yields $(H_{t-1}-H_t)\widehat{S}_t$ on the bank account, where $\widehat{S}_t$ encodes the price achieved per unit of security, given the temporary price impact of trading. Temporary price impact is modeled by the function $\mathfrak g$, writing $\widehat{S}_t= S_t - \mathfrak g(H_{t-1}-H_t)$, where $S_t$ would denote the stock price at time $t$ where there is no trading in the market. We set $H_{-1}=X$ and $H_T=0$. After liquidation the amount on the bank account $V(H)$ is given by
\begin{align*}
V(H)
&:=
\sum_{t=0}^T (H_{t-1}-H_{t}) \widehat{S}_t\\
&=
\sum_{t=0}^T (H_{t-1}-H_{t}) S_t - \sum_{t=0}^T (H_{t-1}-H_{t})\mathfrak g(H_{t-1}-H_t) \\
&=
XS_0 + \sum_{t=0}^{T-1} H_t(S_{t+1} - S_{t}) - \sum_{t=0}^T (H_{t-1}-H_{t})\mathfrak g(H_{t-1}-H_t).
\end{align*} 
We assume that the function $\mathfrak g$ is convex, lower-semicontinuous and nondecreasing, and satisfies $\mathfrak g(0)=0$. 
Moreover, we assume that $\mathfrak g(x)<\infty$ for all $x\geq0$.
Then the above defines a concave market model. The first term above does not depend on the strategy, the second one depends linearly and the last term is obviously concave; indeed, the function $x\mapsto x\mathfrak g(x)$ is, under stated conditions, clearly convex on $x\geq0$. 

Let  $U\colon\R\rightarrow\R\cup\{-\infty\}$ be a nondecreasing, upper-semicontinuous,  concave utility function bounded from above. Define the map $\Psi$ by $\Psi(H):= U(V(H))$. Clearly, $\Psi(\omega,\cdot)$ is concave and upper-semicontinuous as $V(\cdot)$ is concave, upper-semicontinuous and $U$ is concave, upper-semicontinuous and nondecreasing. Boundedness from above follows from the same assumption on $U$, and Borel measurability holds whenever the stock price process $(S_t)$ is Borel.

We assume that the stock price process $(S_t)$ is bounded below by some constant $\gamma\leq 0$, i.e. $S_t(\omega)\geq \gamma$ for all $\omega$. We then show that Assumption~\ref{ass:Psi-MultiP}(4) is satisfied by the model. Define the sequence that encodes constant rate of liquidation
\begin{align*}
h^{\circ}_t:= X\Big(1-\frac{t+1}{T+1}\Big), \quad t\in \{0,\dots,T-1\}.
\end{align*}
and choose $\varepsilon<\frac{X}{3(T+1)}$. Then, it is straightforward to see that for any $(x_t)$ satisfying $\Vert x-h^{\circ}\Vert<\varepsilon$, setting also $x_{-1}=X$ and $x_T=0$, we have
\begin{align*}
V(x)
&=
\sum_{t=0}^T (x_{t-1}-x_t) S_t - \sum_{t=0}^{T-1}  (x_{t-1}-x_{t})\mathfrak g(x_{t-1}-x_{t})\\
&\geq
\frac{5\gamma X}{3(T+1)} -  \frac{5 X}{3(T+1)}\mathfrak g\left(\frac{5 X}{3(T+1)}\right).
\end{align*}
Therefore, $\Psi$ satisfies Assumption~\ref{ass:Psi-MultiP}.
\end{example}
%
\section{One-Period-Model}\label{sec:1-Per}
The key technique of the proof of Theorem~\ref{thm:Maxim-exist-MultiP} is dynamic programming, i.e. backward induction, where one optimizes the strategy over each single time step and then 'glues' the single step strategies together. To that end, it is necessary to analyze the corresponding one-period model first, which we do in this section. In the one-step case one may prove the result in a more general setup than the multi-period case. For that reason we first provide the setup for the one-period case.

\subsection{Setup}\label{subsec:1-Per-setup}
Let $\Omega$ be a Polish space, $\cF$ be the universal completion of the Borel $\sigma$-field $\cB(\Omega)$, 
and $\fP$ be a  possibly nondominated set of probability measures on $\cF$.

We fix a function $\Psi \colon \Omega\times \R^{d} \to \R\cup\{-\infty\}$ and consider the optimization problem
\begin{equation}\label{eq:optimization-problem-1-Per}
\sup_{h \in \R^d} \inf_{P \in \fP} E^P[\Psi(h)].
\end{equation}
Throughout this section, we will work with the following conditions on $\Psi$.
\begin{assumption}\label{ass:1-Per}
The map $\Psi\colon \Omega\times \R^{d} \to \R\cup\{-\infty\}$ satisfies 
\begin{enumerate}
\item[(1)] for every $\omega \in \Omega$, the map $x \mapsto \Psi(\omega,x)$ is concave and upper-semicontinuous;
\item[(2)] there exists a constant $C \in \R$ such that $\Psi(\omega,x)\leq C$ for all $\omega \in \Omega, x \in \R^{d}$;
\item[(3)]  the map $(\omega,x)\mapsto \Psi(\omega,x)$ is 
lower semianalytic;
\item[(4)] the optimization problem is non-trivial, i.e. there exists a strategy $h^\circ\in\cH=\R^d$, such that $\inf_{P\in\fP} E_P[\Psi(h^\circ)]>-\infty.$
\end{enumerate}
\end{assumption}
\begin{remark}
Let us compare Assumption~\ref{ass:1-Per} with Assumption~\ref{ass:Psi-MultiP} of the general multi-period case. First note that points (1), (2) and (3) are the same in both cases. In fact, Assumption~\ref{ass:1-Per}(3) insures that for every $x \in \R^d$ the random variable $\Psi(x)$ is measurable. This can be achieved e.g. by assuming the weaker condition of $(\omega,x)\mapsto \Psi(\omega,x)$ being universally measurable; see \cite[Lemma~7.29, p.174]{BertsekasShreve.78}. However in the multi-period case, lower semianalyticity cannot be relaxed without losing the measurability setting needed to be able to apply crucial measurable selection arguments; see the proof of Proposition~\ref{prop:local-NA} and Proposition~\ref{prop:no-arbitrage-up-time-t-MultiP}.
Assumption~\ref{ass:1-Per}(4) just requires that the optimization problem is well posed, i.e. the value is not equal to $-\infty$ for all $x \in \R^d$. The stronger Assumption~\ref{ass:Psi-MultiP}(4) serves a different purpose; see Remark~\ref{rem:on assumption 2.1}(3).
\end{remark}
We work under the following no-arbitrage condition 
\begin{align}\label{eq:def-NA-1-Per}
\textrm{the set }\quad
\cK := \big\{h\in\R^d\,\big|\,\Psi^\infty(h)\geq 0\quad\fP\textrm{-q.s.}\big\}
\quad \textrm{is linear}.
\end{align}
The main result of this section is the existence of a maximizer for the one-period optimization problem \eqref{eq:optimization-problem-1-Per}.
\begin{theorem}\label{thm:1-Per}
Let the no-arbitrage condition \eqref{eq:def-NA-1-Per} and Assumption~\ref{ass:1-Per} hold. There exists a strategy $\widehat{h} \in \R^d$ such that
\begin{equation}\label{eq:thm:1-Per}
\inf_{P\in\fP} E_P[\Psi(\widehat h)]= \sup_{h\in\R^d}\inf_{P\in\fP} E_P[\Psi(h)],
\end{equation}
i.e. the  optimization problem admits an optimizer.
\end{theorem}
\subsection{Proof of the One-Period Optimization Problem}\label{subsec:proof-1-Per}

In the one-step case, one can  
prove the existence of a maximizer just by applying classical arguments from convex analysis. Assumption~\ref{ass:1-Per} is in force throughout this section. One of the key elements of the proof is the following reformulation of \cite[Theorem~9.2, p.75]{Rockafellar.97}.
\begin{proposition}\label{prop:thm9.2R}
Let $f:\R^n\rightarrow\R\cup\{-\infty\}$ be concave, upper-semicontinuous and proper, and let $A:\R^n\rightarrow\R^m$ be linear. If the set $\{z \in \R^n \,|\, Az=0,\, f^\infty(z)\geq0 \}$ is linear, then the function
\begin{equation*}
g(y):=\sup\{f(x)\,|\,x\in\R^n,\,Ax=y\}
\end{equation*}
is concave, proper and upper-semicontinuous. Moreover,
\begin{equation*}
g^\infty(y)=\sup\{f^\infty(x)\,|\,x\in\R^n,\,Ax=y\}.
\end{equation*}
Furthermore, for each $y$ such that $g(y)> -\infty$, the supremum in the definition of $g$ is attained.
\end{proposition}
%
%
Consider the function  
\begin{equation*}
\Phi\colon h\mapsto\inf_{P\in\fP} E^P[\Psi(h)].
\end{equation*}
Observe that  $\Phi$ is concave and upper-semicontinuous as an infimum of concave upper-semicontinuous functions; use Assumption~\ref{ass:1-Per}(1,2) and Fatou's lemma. It is also proper, i.e. not identically equal to $-\infty$, by Assumption~\ref{ass:1-Per}(4). Moreover, we have the following.
\begin{lemma}\label{le:psi-phi-infty}
Let $\Psi:\Omega\times \R^{d} \to \R\cup\{-\infty\}$ satisfy Assumption~\ref{ass:1-Per}. Then
\begin{equation*}
\Phi^\infty(h)= \inf_{P \in \fP} E^P[\Psi^\infty(h)]
\end{equation*}
for all $h\in\R^d$.
\end{lemma}
\begin{proof}
By concavity of $\Psi(\omega,\cdot)$, the sequence  $n\mapsto\frac1n\big(\Psi(\omega,nh+h^\circ)-\Psi(\omega,h^\circ)\big)$ is pointwise decreasing for every $\omega$, hence monotone convergence yields
\begin{align*}
\Phi^\infty(h)
&= \inf_{n \in \N} \frac{1}{n}
\Big(
\inf_{P\in \fP} E^P[\Psi(nh+h^\circ)]-\inf_{P\in \fP} E^P[\Psi(h^\circ)]
\Big)\\
&=\inf_{n \in \N} \frac{1}{n}
\Big(
\inf_{P\in \fP} \big(E^P[\Psi(nh+h^\circ)-\Psi(h^\circ)]+E^P[\Psi(h^\circ)]\big)-\inf_{P\in \fP} E^P[\Psi(h^\circ)]
\Big)\\
&\geq\inf_{n \in \N} \frac{1}{n}
\inf_{P\in \fP} \big(E^P[\Psi(nh+h^\circ)-\Psi(h^\circ)]\big)\\
&= 
\inf_{P\in \fP} \lim\limits_{n \to \infty} \frac{1}{n} \big(E^P[\Psi(nh+h^\circ)-\Psi(h^\circ)]\big)\\
&= \inf_{P \in \fP} E^P[\Psi^\infty(h)]
\end{align*}
For the reverse inequality, we use  Assumption~\ref{ass:1-Per}(2) and come back to the second line of the above calculation to get
\begin{align*}
\Phi^\infty(h)
&=\inf_{n \in \N} \frac{1}{n}
\Big(
\inf_{P\in \fP} \big(E^P[\Psi(nh+h^\circ)-\Psi(h^\circ)]+E^P[\Psi(h^\circ)]\big)-\inf_{P\in \fP} E^P[\Psi(h^\circ)]
\Big)\\
&\leq\inf_{n \in \N} \frac{1}{n}\Big(
\inf_{P\in \fP} \big(E^P[\Psi(nh+h^\circ)-\Psi(h^\circ)]\big) + C-\inf_{P\in \fP} E^P[\Psi(h^\circ)] \Big)\\
&= \inf_{P \in \fP} E^P[\Psi^\infty(h)].
\end{align*}
\end{proof}

\begin{proof}[Proof of Theorem~\ref{thm:1-Per}]
The mapping $\Psi(\omega,\cdot)$ is bounded from above by a constant $C$ for each $\omega\in\Omega$. Thus $\Psi^\infty(\omega,\cdot)\leq 0$ for each $\omega$. From Lemma~\ref{le:psi-phi-infty} we know that $\Phi^\infty(h)= \inf_{P \in \fP} E^P[\Psi^\infty(h)]$ and thus also $\Phi^\infty(h)\leq0$ for all $h\in\R^d$.

Observe  that $\Phi^\infty(h) = 0$ if and only if $h\in\cK$. Therefore, by the no-arbitrage condition \eqref{eq:def-NA-1-Per}, we see that the conditions of Proposition~\ref{prop:thm9.2R} are fulfilled for the linear map $A:\R^d \to \{0\}, \ x \mapsto 0$ and $f\equiv \Phi$. Thus,  Proposition~\ref{prop:thm9.2R} yields  the existence of a maximizer $\widehat h$, as (by choosing $y\equiv0$)
\begin{align*}
\sup_{h \in \R^d} \Phi(h)
= 
\sup\{\Phi(h)\,|\, Ah=0\}
=
\sup_{h \in \R^d} \inf_{P \in \fP} E^P[\Psi(h)]
\geq\inf_{P \in \fP} E^P[\Psi(h^\circ)]
>-\infty.
\end{align*}
\end{proof}

\begin{remark}\label{rem:proof-thm-1-Per}
Notice that in the above proof we did not require any structural properties of the measurable space $(\Omega,\cF)$ or of the set $\fP$ of probability measures. The only important element of the setup is concavity of the map $\Psi$. In particular, Theorem~\ref{thm:1-Per} remains valid for any measurable space $(\Omega, \cF)$ when replacing Assumption~\ref{ass:1-Per}(3) by the assumption  that the map $(\omega,x) \mapsto \Psi(\omega,x)$ is $\cF \otimes \cB(\R^d)$-measurable.
\end{remark}
%

%
\section{Multi-Period-Model}\label{sec:proof-multi-period}
The key idea of the multi-period case is to adapt the techniques of dynamic programming principle, developed in~\cite{Evstigneev.76}, to the robust framework. We will use the setup and measurability techniques developed in \cite{BouchardNutz.13} and \cite{Nutz.13util}.

From this point on, Assumption~\ref{ass:Psi-MultiP} is in force as well as the no-arbitrage condition $\NA(\fP)$ given by
\begin{equation*}
\mbox{the set }
\cK:=\{H \in \cH \, | \, \Psi^\infty(H)\geq 0 \ \fP\mbox{-q.s.}\} \, \mbox{ is linear.}
\end{equation*} %

The basic idea of dynamic programming is to reduce the maximization over the set of strategies $\cH$ to one-step maximization we encountered in the previous section. Precisely, using the notation where for $\omega^t \in \Omega^t$  and $\tilde{\omega}\in \Omega_1$, $\omega^t\otimes_t\tilde{\omega}$ stands for the pair $(\omega^t,\tilde{\omega})\in \Omega^{t+1}$, we define the following sequences of maps: set $\Psi_{T}:=\Psi$ and for $t=T-1,\dots,0$ and $\omega^t \in \Omega^t$ define
\begin{align*}
\Phi_t(\omega^t,x^{t+1})
&:=\inf_{P \in \fP_t(\omega^t)}E^P[\Psi_{t+1}(\omega^t\otimes_t\,\cdot,x^{t+1})],\\
\widetilde \Psi_t(\omega^t,x^{t})
&:=\sup_{\tilde{x}\in \R^d} \Phi_t(\omega^t,x^{t},\tilde{x}),\\
\Psi_t(\omega^t,\cdot)
&:=  \cl \widetilde \Psi_t(\omega^t, \cdot),
\end{align*}
where by $\cl \widetilde \Psi_t(\omega^t, \cdot)$ we denoted the upper-semicontinuous hull of the function $x\mapsto\widetilde \Psi_t(\omega^t,x)$. The upper-semicontinuous hull of a function $f\colon\R^n\to [-\infty,\infty]$ is the smallest upper-semicontinuous function (not necessarily finite) minorized by $f$; it is the function whose hypograph is the closure in $\R^n \times \R$ of the hypograph of $f$, see \cite[p.52]{Rockafellar.97}. 

\begin{remark}
Let us give a rough sketch of an argument why one would consider the recursion above. We will restrict our attention to the case with $T=2$. The optimization problem we are considering is
$$
\sup_{H_0,H_1}\inf_{P\in\fP}E^P[\Psi_2(H_0,H_1)]
=
\sup_{H_0}\bigg[\sup_{H_1}\inf_{P\in\fP}E^P[\Psi_2(H_0,H_1)]\bigg].
$$
After we evaluated the expression in the brackets, optimization over $H_0$ will follow the argument we provided for the one-step case. Let us, therefore, concentrate on the optimization over $H_1$. First, calculate the conditional expectation of $\Psi_2(H_0,H_1)$ given $\cF_1$. By the definition of the probability measure as $P = P_0\otimes P_1$, we know that 
\begin{align*}
E^{P_0\otimes P_1}[\Psi_2(H_0,H_1)|\cF_1](\omega_1)
&=
\int_{\Omega_2}\Psi_2(\omega_1\otimes\omega_2,H_0,H_1)P_1(\omega_1,d\omega_2)
\\&=
E^{P_1(\omega_1)}[\Psi_2(\omega_1\otimes_1\cdot,H_0,H_1(\omega_1))],
\end{align*}
where the second equality is just change in notation. Hence, $\Phi_1$ is just
$$
\Phi_1(\omega_1,x_0,x_1)= \Phi_1(\omega_1,x^2)
=
\inf_{P_1 \in \fP_1(\omega_1)}E^{P_1}[\Psi_2(x^{2})|\cF_1](\omega_1),
$$
where the versions of the conditional expectation are defined via kernels $P_1\in\fP_1(\omega_1)$. Coming back to the minimization over $H_1$, first use the tower property: for any strategy $H\in\cH$ we have 
$$
\inf_{P \in \fP}E^P[\Psi_2(H_0,H_1)]
=
\inf_{P_0 \in \fP_0}E^{P_0}\bigg[\inf_{P_1 \in \fP_1(\omega_1)}E^{P_1}[\Psi_2(\omega_1\otimes_1\,\cdot,H_0,H_1)]\bigg]
$$
by the decomposability property. However, one would need to show that the expression in the expectation is measurable. In particular, the above equation implies that
\begin{equation*}
\sup_{H_1}\inf_{P\in\fP}E^P[\Psi_2(H_0,H_1)] =
\sup_{H_1} \inf_{P \in \fP}E^{P}[\Phi_1(\cdot,H_0,H_1)].
\end{equation*}

It remains to show that
$$
\sup_{H_1}\inf_{P \in \fP}E^{P}[\Phi_1(\cdot,H_0,H_1)]
=
\inf_{P_0 \in \fP_0}E^{P_0}[\Psi_1(\cdot,H_0)].
$$
The inequality $\leq$ is easy to see. Indeed, we have $\Phi_1(\omega,H_0,H_1)\leq \Psi_1(\omega,H_0)$ for each $\omega$, $H_0$. To see the converse, it is enough to prove that the supremum in the definition of $\Psi_1$ is attained by some $\cF_1$-measurable random variable. 

This formally derives the recursion we are considering for the  maximization problem.
\end{remark}

Of course, the remark above does not constitute a proof to our main theorem, but merely a sketch of an argument of why the recursion we are observing makes sense. Indeed, the main part of the argument will be proving measurability and attainment of maximizers in the one-step case.

\begin{remark}\label{rem:on assumption 2.1}
Before going to the more difficult part of the proof, let us make a few simple observations about the recursion defined above. The numbering refers to the numbering in the Assumption~\ref{ass:Psi-MultiP}.
\begin{enumerate}
\item[{\bf(1a)}]
It follows directly from the definitions that the mappings $\Phi_t$ and $\Psi_t$ are concave. Indeed, we assume that $\Psi_T=\Psi$ is concave, and \cite[Proposition~2.9(b), p.43]{RockafellarWets} implies that $\Phi_t$ is concave as soon as $\Psi_{t+1}$ is. Moreover, \cite[Proposition~2.22(a), p.51]{RockafellarWets} implies that $\widetilde \Psi_t$ 
and hence also $\Psi_t$ is concave as soon as $\Phi_t$ is. This holds for each $\omega^t\in\Omega^t$.
\item[{\bf(1b)}]
For each $\omega^t$, the map $\Phi_t$ is upper-semicontinuous in the second variable as long as $\Psi_{t+1}$ is. The argument, which uses Fatou lemma and the uniform upper bound, is given in the one-step case.
Moreover, upper-semicontinuity of $\Psi_t$ follows directly from its definition as an upper-semicontinuous hull of the proper  concave function $\widetilde{\Psi}_t$. In Proposition~\ref{prop:Psi-MultiP-OneP}(i), we will show that also the function  $\widetilde{\Psi}_t$ is upper-semicontinuous for $\fP$-q.e. $\omega^t$. This will then imply that $\widetilde\Psi_t$ coincides $\fP$-q.s. with $\Psi_t$. The proof will be based on Proposition~\ref{prop:thm9.2R} and a local no-arbitrage condition which we will introduce in Definition~\ref{def:NA-t}.
\item[{\bf(2)}] 
Boundedness from above by a constant $C\in\R$ is obvious: take the same constant $C$ as in Assumption~\ref{ass:Psi-MultiP}(2) for the mapping $\Psi$. 
\item[{\bf(3)}]
We will prove in Lemma~\ref{le:ass-psi-phi-multiP} that $\Phi_t$ is lower semianalytic if $\Psi_{t+1}$ is. Semianalyticity of $\Psi_t$ is based on a density argument; see Lemma~\ref{le:dense-MultiP}. The problem arises as the supremum of lower semianalytic functions need not be lower semianalytic. Passing to a supremum of a countable number of functions, i.e. maximizing over $\Q^d$ instead of $\R^d$ in the definition of  $\widetilde\Psi_t$, preserves this measurability property. This is the reason for condition~(4) in Assumption~\ref{ass:Psi-MultiP}.
\item[{\bf(4)}]
It is easy to see that the condition~\eqref{eq:interior} is satisfied for the mapping $\Phi_t$ as long as it is satisfied for the mapping $\Psi_{t+1}$
To see that the mapping $\widetilde\Psi_t$ needs to satisfy~\eqref{eq:interior} when the condition is satisfied by $\Phi_t$, it is enough to notice that $\widetilde\Psi_t(\omega^t,x^t)\geq \Phi_t(\omega^t,x^t,h^{\circ}_{t})$ for each $(\omega^t,x^t)\in\Omega^t\times\R^{dt}$, where $h^{\circ}_{t}$ is the $t+1$-entry of the vector $h^{\circ}=(h^{\circ}_0,\dots,h^{\circ}_{T-1})$ defined in Assumption~\ref{ass:Psi-MultiP}(4). Then, as $\Psi_t\geq \widetilde\Psi_t$ by definition of the upper-semicontinuous hull,  $\Psi_t$  satisfies~\eqref{eq:interior}.
\end{enumerate}
\end{remark}

\begin{remark}
We just argued in Remark~\ref{rem:on assumption 2.1} that the mapping $x^t\mapsto\widetilde \Psi_t(\omega^t,x^t)$ is concave for each $t$ and all $\omega^t\in\Omega^t$ and  that $\Psi_t(\omega^t,h^{\circ,t})>-c$ for all $\omega^t$, where $h^{\circ,t}\in \R^{dt}$ denotes the restriction of $h^\circ$ defined in Assumption~\ref{ass:Psi-MultiP}(4) to the first $t$ entries. Therefore, \cite[Theorem~2.35, p.59]{RockafellarWets} provides the identity 
$$
\Psi_t(\omega^t,x^t) = \lim_{\lambda\nearrow 1}\widetilde \Psi_t(\omega^t, \lambda x^t+(1-\lambda)h^{\circ,t}).
$$
\end{remark}

We first need to show that if the mapping $\Psi$ satisfies the Assumption~\ref{ass:Psi-MultiP} then also all of the mappings $\Psi_t$ do for each $t =0,\dots,T-1$. This has to be deduced from the dynamic programming recursion. For that we first need some more terminology and notation. For each $t \in \{1,\dots,T\}$, denote by $\cH^t$ the set of all $\F$-adapted, $\R^d$-valued processes $H^t:=(H_0,\dots,H_{t-1})$; these are just restrictions of strategies in $\cH$ to the first $t$ time steps. Define the no-arbitrage condition up to time $t$, denoted by $\mbox{NA}(\fP)^t$, for the mappings $(\Psi_t)$ in the natural way, by saying
\begin{align*}
\mbox{the set } \ \cK^t:=\{H^t\in \cH^t \, | \, \Psi^\infty_t(H^t)\geq 0 \ \fP\mbox{-q.s.} \} \ \mbox{is linear.}
\end{align*}
Condition $\mbox{NA}(\fP)^t$ is a statement about a set of strategies and as such cannot yet be used to prove things that we need it for. What we need is a local version of the no-arbitrage condition. 
\begin{definition}\label{def:NA-t}
For each $t \in \{0,\dots,T-1\}$ and $\omega^t \in \Omega^t$ define a set 
\begin{equation}\label{eq:def-K-t-local}
K_t(\omega^t):=\{h \in \R^{d} \,|\, \Psi^\infty_{t+1}(\omega^t\otimes_t\cdot,0,\dots,0,h)\geq 0 \ \  \fP_{t}(\omega^t)\mbox{-q.s.}\, \}.
\end{equation}
We say that condition $\mbox{NA}_t$ holds if
\begin{equation*}
\mbox{the set } \,\big\{\omega^t \in \Omega^t \, \big|\, K_t(\omega^t) \mbox{ is linear} \big\} \  \mbox{ has }\fP\mbox{-full measure.}
\end{equation*}
\end{definition}
To make a simple observation: let $\widehat h$ be a universally measurable selection of $K_t$. Then the strategy $(0,\ldots,0,\widehat h)\in\cK^{t+1}$, where there are $t$ zeros in the previous expression. Thus, it is clear that linearity of $K_t$ is necessary for the no-arbitrage condition $\mbox{NA}(\fP)^{t+1}$. 
\begin{proposition}\label{prop:local-NA}
Let $t\in\{0,\ldots,T-1\}$. Assume that $\Psi_{t+1}$ satisfies Assumption~\ref{ass:Psi-MultiP}. If $\cK^{t+1}$ is linear, 
then the local no-arbitrage condition $\NA_t$ holds.
\end{proposition}
Having this result at hand, one can proceed as in the one-step case. The following is a direct consequence.
\begin{proposition}\label{prop:Psi-MultiP-OneP}
Let $t \in \{0,\dots,T-1\}$. Assume that 
$\Psi_{t+1}$ satisfies Assumption~\ref{ass:Psi-MultiP} and that $\mathrm{NA}_t$ holds. Then
\begin{enumerate}
\item the map $x^t\mapsto\widetilde\Psi_t(\omega^t,x^t)$ is upper-semicontinuous $\fP$-q.s.; in particular, it coincides with $x^t\mapsto\Psi_t(\omega^t,x^t)$ for $\fP$-quasi every $\omega^t$.
\item for every $H^t \in \cH^t$ there exists an 
$\cF_t$-measurable mapping $\widehat{h}_t:\Omega^t\to\R^d$ such that
\begin{equation*}
\Phi_t(\omega^t,H^t(\omega^t),\widehat{h}_t)
=\Psi_t(\omega^t,H^t(\omega^t)) 
\qquad
\mbox{for }\fP\mbox{-q.e. }\omega^t \in \Omega^t.
\end{equation*}
\end{enumerate}
\end{proposition}
%
%
%
%
The  important step toward the proof of our main result is the observation that sets $\cK^t$ behave well under the dynamic programming recursion. 
\begin{proposition}\label{prop:no-arbitrage-up-time-t-MultiP}
Let $t\in \{0,\ldots,T-1\}$ and let $\Psi_{t+1}$ satisfy Assumption~\ref{ass:Psi-MultiP}. If $\cK^{t+1}$ is linear, then so is $\cK^t$.
\end{proposition}
The proof of the following Proposition 
will be done by backward induction.  
\begin{proposition}\label{prop:induction}
For any $t \in \{0,\dots,T-1\}$, the function $\Psi_{t+1}$ satisfies Assumption~\ref{ass:Psi-MultiP} and the local no-arbitrage condition $\mathrm{NA}_t$ holds.
\end{proposition}
The proofs of Propositions~\ref{prop:local-NA}--\ref{prop:induction} will be given in the next subsection.

%
%
\subsection{Proofs of Propositions~\ref{prop:local-NA}--\ref{prop:induction}}\label{subsec:global-local-prop-proof}
This is the technical part of this paper, and therefore is divided into several lemmas.  We first start with a  useful lemma providing the relation between $\Psi_{t+1}$ and $\Phi_t$.
\begin{lemma}\label{le:ass-psi-phi-multiP}
Let $t \in \{0,\dots,T-1\}$. If $\Psi_{t+1}$ satisfies Assumption~\ref{ass:Psi-MultiP}, then so does $\Phi_{t}$ and 
for all $(\omega^t,x^{t+1})\in \Omega^t\times \R^{d(t+1)}$ we have
\begin{equation}\label{eq:psi-phi-infty-multiP}
\Phi^\infty_t(\omega^t,x^{t+1})=\inf_{P \in \fP_t(\omega^t)}E^P[\Psi^\infty_{t+1}(\omega^t\otimes_t\,\cdot,x^{t+1})].
\end{equation}
\end{lemma}
\begin{proof}
Conditions (1), (2) and (4) are clear by definition and were argued in Remark~\ref{rem:on assumption 2.1}. To see that $(\omega^t,x^{t+1})\mapsto\Phi_t(\omega^t,x^{t+1})$ is lower semianalytic, we first recall  that the map $(\omega^{t+1},x^{t+1})\mapsto\Psi_{t+1}(\omega^{t+1},x^{t+1})$ is lower semianalytic by assumption. 
Also, the map
\begin{equation*}
\fM_1(\Omega)\times\Omega^t\times\Omega_1 \times \R^{d(t+1)}  \to \overline{\R}, \quad (P,\omega^t,\tilde\omega,x^{t+1})\mapsto \Psi_{t+1}(\omega^t\otimes_t\tilde\omega,x^{t+1})
\end{equation*}
is lower semianalytic as it is independent of the variable $P$. Consider the Borel measurable stochastic kernel $\kappa$ on $\Omega_1$ given $\fM_1(\Omega)\times\Omega^t\times \R^{d(t+1)}$ defined by
\begin{equation*}
\big((P,\omega^t,x^{t+1}),A\big)\mapsto 
\kappa(A\,|\,P,\omega^t,x^{t+1}):=P[A];
\end{equation*}
Borel measurability of the kernel follows from~\cite[Proposition~7.26, p.134]{BertsekasShreve.78} and~\cite[Corollary~7.29.1, p.144]{BertsekasShreve.78}. 
Then, applying~\cite[Proposition~7.48, p.180]{BertsekasShreve.78} to $\kappa$, we obtain that
\begin{equation}\label{eq:lsa-kernel}
\fM_1(\Omega)\times\Omega^t \times \R^{d(t+1)}  \to \overline{\R}, \quad (P,\omega^t,x^{t+1})\mapsto E^P[\Psi_{t+1}(\omega^t\otimes_t\,\cdot,x^{t+1})]
\end{equation}                  
is lower semianalytic. By assumption, the graph of $\fP_t$ is analytic. Therefore, we deduce from \cite[Lemma~7.47, p.179]{BertsekasShreve.78} that
\begin{equation*}
(\omega^t,x^{t+1})\mapsto \inf_{P \in \fP_t(\omega^t)}E^P[\Psi_{t+1}(\omega^t\otimes_t\,\cdot,x^{t+1})]=\Phi_t(\omega^t,x^{t+1})
\end{equation*} 
is lower semianalytic. Finally, by the same arguments as in the proof of Lemma~\ref{le:psi-phi-infty}, we see directly that \eqref{eq:psi-phi-infty-multiP} holds true.
\end{proof}
%
%
Before we can start 
with the proof of Proposition~\ref{prop:local-NA}, we need to see that the set valued map $K_t(\omega^t)$ has some desirable properties.
\begin{lemma}\label{le:K-t-local-nice}
Let $t \in \{0,\dots,T-1\}$. Assume that $\Psi_{t+1}$ satisfy Assumption~\ref{ass:Psi-MultiP}. Then the set-valued map $K_t$ defined in \eqref{eq:def-K-t-local} is a closed, convex, $\cF_t$-measurable correspondence and the set $\big\{\omega^t \in \Omega^t \, \big|\, K_t(\omega^t) \mbox{\rm{ is linear}} \big\}\in\cF_t$.
\end{lemma}
\begin{proof}
As $\Psi_{t+1}^\infty$ is concave, positively homogeneous and upper-semicontinuous in $x^t$, $K_t(\omega^t)$ is a closed valued convex cone for every $\omega^t$.
Observe that
\begin{align*}
K_t(\omega^t)
=
\{x_t\in\R^d\,|\,\Phi^\infty_t(\omega^t,0,\dots,0,x_t)\geq0\}.
\end{align*}
By Lemma~\ref{le:ass-psi-phi-multiP} and Remark~\ref{rem:ass-normal-multiP}, $\Phi_t$ is a concave $\cF_t$-normal integrand, hence so is $\Phi_t^\infty$.
Thus, the set valued map $K_t$ is an $\cF_t$-measurable correspondence; see  \cite[Proposition~14.33, p.663]{RockafellarWets} and \cite[Proposition~14.45(a), p.669]{RockafellarWets}. 

Finally, from $K_t(\omega^t)$ being a convex cone, we get that
\begin{equation*}
\big\{\omega^t \in \Omega^t \, \big|\, K_t(\omega^t) \mbox{\rm{ is linear}}\big\} 
=
\big\{\omega^t \in \Omega^t \, \big|\, K_t(\omega^t)=- K_t(\omega^t)\big\}.
\end{equation*}
By \cite[Theorem~14.5(a), p.646]{RockafellarWets}, $K_t$ admits a Castaing representation $\{x_n\}$. Then, we see that 
\begin{equation*}
\big\{\omega^t \in \Omega^t \, \big|\, K_t(\omega^t) \mbox{\rm{ is linear}} \big\} 
= 
\bigcap_{n\in\N}\{\omega^t \,|\, -x_n(\omega^t)\in K_t(\omega^t)\}.
\end{equation*}
That the latter set is $\cF_t$-measurable now follows from \cite[Definition~14.3(c), p.644]{RockafellarWets} and \cite[Proposition~14.11(c), p.651]{RockafellarWets}.
\end{proof}
%
%
%
%
Now we prove that if $\Psi_{t+1}$ satisfies Assumption~\ref{ass:Psi-MultiP}, linearity of $\cK^{t+1}$ implies $\mathrm{NA}_t$.
\begin{proof}[Proof of Proposition~\ref{prop:local-NA}]
Assume that $\mathrm{NA}_t$ does not hold. Then by definition, there exists a probability measure in $\fP$ with its restriction to $\Omega^t$ denoted by $P^t$ such that
the  complement of the set
\begin{equation*}
G_t:=\big\{\omega^t \in \Omega^t \, \big|\, K_t(\omega^t)=- K_t(\omega^t)\big\}
\end{equation*}
satisfies $ P^t[G^c_t]>0$.
We claim that there is a strategy of the form $H^{t+1} = (0,\ldots,0,h_{t})$ and a measure $\widetilde{P}\in\fP$
such that $H^{t+1}\in\cK^{t+1}$, but $\widetilde{P}[\Psi^\infty_{t+1}(-H^{t+1})<0]>0$, i.e. $-H^{t+1}\not\in\cK^{t+1}$. The measure $\widetilde{P}$, restricted to $\Omega^{t+1}$, will be defined as $P^{t+1}=P^t\otimes P_t$ for some selection $P_t$ of $\fP_t$.

\textsc{Step 1:} We prove that there is a Borel measurable set-valued map $ K^{P^t}_t\colon\Omega^t\rightrightarrows\R^d$ that coincides with $K_t$ for $P^t$-a.a. $\omega^t$.

Let $\{x_n\}$ be the Castaing representation of $K_t$.
By~\cite[Lemma~7.27, p.173]{BertsekasShreve.78}  we can modify each of these universally measurable selections $x_n$ on a $P^t$-nullset to get a almost sure selections $x^{P^t}_n$ that are Borel measurable. Define a new set-valued map 
\begin{align*}
K^{P^t}_t(\omega^t) = \overline{\{ x^{P^t}_n(\omega^t)\,|\,n\in\N\}}.
\end{align*}
The set-valued map $K^{P^t}_t$ is Borel measurable by definition and $ K^{P^t}_t(\omega^t)        =K_t(\omega^t)$ for $P^t$-a.a. $\omega^t$; this follows from \cite[Proposition~14.11, p.651]{RockafellarWets} and \cite[Proposition~14.2, p.644]{RockafellarWets}.

\textsc{Step 2:} Define the set $\cS_t\subset\Omega^t\times \R^{d}\times \fM_1(\Omega_1)$ by
\begin{equation*}
\cS_t
:=\Big\{(\omega^t,h,P)\,\Big|\,h\in  K^{P^t}_t(\omega^t), \, P \in \fP_t(\omega^t), \, E^P\big[\Psi^\infty_{t+1}(\omega^t\otimes_t\cdot,0,\dots,0,-h)\big]<0\Big\}.
\end{equation*}
We claim that this set $\cS_t$ is analytic. 

To see this, write $\cS_t$ as an intersection of three sets:
\begin{align*}
B_1 &:= \big\{(\omega^t,h)\,|\,h\in  K^{P^t}_t(\omega^t)\big\}\times\fM_1(\Omega_1),\\
B_2 &:= \big\{(\omega^t,P)\,|\,P\in \fP_t(\omega^t)\big\}\times\R^d,\\
B_3 &:= \big\{(\omega^t,h,P)\,|\,E^P[\Psi^\infty_{t+1}(\omega^t\otimes_t\cdot,0,\dots,0,-h)]<0\big\},
\end{align*}
and show that each of those is analytic.

The set $B_1$ is Borel, as it is just a product of $\fM_1(\Omega_1)$ and the graph of $ K^{P^t}_t(\omega^t)$, which is Borel, see \cite[Theorem~14.8, p.648]{RockafellarWets}. 

The set $B_2$ is analytic being the product of $\R^d$ and the graph of $\fP_t$, which is analytic by assumption.

To show that $B_3$ is analytic, use the assumption that  $\Psi_{t+1}$ is a lower semianalytic map. By~\cite[Lemma~7.30(2), p.177]{BertsekasShreve.78} also the map $\Psi_{t+1}^\infty$ is; it is defined as a limit of lower semianalytic functions. By the same argument as in Lemma~\ref{le:ass-psi-phi-multiP}, we see that the set 
\begin{equation*}
\{(\omega^t,h_0,\dots,h_t,P)\in \Omega^t\times \R^{d(t+1)}\times \fM_1(\Omega_1)\,|\,E^P[\Psi^\infty_{t+1}(\omega^t\otimes_t\cdot,-h_0,\dots,-h_t)]<0\}
\end{equation*}
is analytic. Denote the above set by $\widetilde{B}_3$.
The projection 
\begin{align*}
\Pi\colon\Omega^{t+1}\times\R^{d(t+1)}\times \fM_1(\Omega_1)&\to \Omega^t\times \R^{d}\times \fM_1(\Omega_1) 
\\ 
(\omega^t,x_0,\dots,x_t,P)&\mapsto (\omega^t,x_t,P)
\end{align*} 
is continuous, and thus Borel. We deduce from \cite[Proposition~7.40, p.165]{BertsekasShreve.78} that the set
\begin{equation*}
B_3=\{(\omega^t,h,P)\in \Omega^t\times \R^d\times \fM_1(\Omega_1)\,|\,E^P[\Psi^\infty_{t+1}(\omega^t\otimes_t\cdot,0,\dots,0,-h)]<0\}
\end{equation*}
is analytic, as
\begin{equation*}
B_3=\Pi\Big(\widetilde{B}_3\cap\big(\Omega^t\times(\{0\}^{dt}\times\R^d)\times \fM_1(\Omega_1)\big) \Big).
\end{equation*} 

\textsc{Step 3:} The desired strategy can be obtained from the selection of $\cS_t$.

Define the set 
$$
\proj\cS_t=\big\{\omega^t\,\big|\,\cS_t\cap(\{\omega^t\}\times\R^d\times\fM_1(\Omega_1))\not=\varnothing \big\},
$$
which is just the projection of the set $\cS_t$ onto the first coordinate. Let us first show that the sets $\proj\cS_t$ and $G_t^c$ are equal up to a $P^t$ nullset. Recall that the probability measure $P^t$ was chosen at the beginning of the proof. Then, by definition of the sets $G^c_t$ and $\cS_t$ and as $K^{P^t}_t=K_t \ P^t$-a.s., 
we have $P^t[\proj\cS_t]=P^t[G_t^c]>0$. 

From Step~2, we know that $\cS_t$ is analytic. Therefore, the Jankov-von Neumann theorem  \cite[Proposition~7.49, p.182]{BertsekasShreve.78} implies the existence of a universally measurable map 
$\omega^t \mapsto (h_t(\omega^t),P_t(\omega^t))$ 
such that 
$(\omega^t,h_t(\omega^t),P_t(\omega^t)) \in \cS_t$ for all $\omega^t \in \proj\cS_t$. On the universally measurable set 
$\{h_t\notin K_t\}\subseteq \Omega^t$, we set $h_t:=0\in\R^d$ to guarantee that $h_t(\omega^t) \in K_t(\omega^t)$ for every $\omega^t\in\Omega^t$.  In the same way, we can define $P_t(\cdot)$ to be any measurable selector of $\fP_t(\cdot)$ on $\{h_t\notin K_t\}$. Recall that $\{h_t\notin K_t\}\cap \proj\cS_t$ is a $P^t$-nullset.

Finally, we claim for the strategy $H^{t+1}:=(0,\ldots,0,h_t)\in \cH^{t+1}$ that $H^{t+1}\in\cK^{t+1}$, but $-H^{t+1}\not\in\cK^{t+1}$. To see this, observe first that as $h_t(\omega^t) \in K_t(\omega^t)$ for all $\omega^t$, we have by definition 
\begin{equation*}
\Psi^\infty_{t+1}(\omega^t\otimes_t\cdot,H^{t+1}(\omega^t))\geq 0 \quad \fP_t(\omega^t)\mbox{-q.s.}
\end{equation*}
for all $\omega^t$. For every $\bar P \in  \fP$ denote its restriction to $\Omega^t$ by $\bar P^t$, i.e. $\bar P^t:=\bar P|_{\Omega^{t}}$. By the definition of the set $\fP$ we have $\bar P^{t+1}=\bar P^t\otimes \bar P_t$ for some selector $\bar P_t\in\fP_t$. By Fubinis theorem we get that every $\bar P\in \fP$ satisfies
\begin{equation*}
\bar P[\Psi^\infty_{t+1}(H^{t+1})\geq 0]=E^{\bar P^t(d\omega^t)}\Big[\bar P_t(\omega^t)\big[\Psi^\infty_{t+1}(\omega^t\otimes_t\cdot,H^{t+1}(\omega^t))\geq 0\big]\Big]=1,
\end{equation*} 
which proves that $H^{t+1}\in\cK^{t+1}$. To see that $-H^{t+1}\not\in\cK^{t+1}$,
define the measure $\widetilde{P} \in \fP$ by
\begin{equation*}
\widetilde{P}:=P^t\otimes P_t\otimes \widetilde P_{t+1}\otimes\dots\otimes \widetilde  P_{T-1}, 
\end{equation*}
where $P^t$ is the measure introduced at the beginning of the proof, the kernel $P_t$ is the one selected from $\cS_t$ above and $\widetilde P_{s} \in \fP_s$ are any selections of $\fP_s$ for $s:=t+1,\dots T-1$. Recalling that $\Psi^\infty_{t+1}\leq 0$, by definition of $h$ we get
\begin{equation*}
E^{\widetilde P}[\Psi^\infty_{t+1}(-H^t)]
=
E^{P^t(d\omega^t)}\Big[E^{P_t(\omega^t)}\big[\Psi^\infty_{t+1}(\omega^t\otimes_t\cdot,-H^t(\omega^t))\big]\Big]<0,
\end{equation*}
as $\widetilde P[\proj \cS_t] >0$. 
Hence $-H^{t+1}\not\in\cK^{t+1}$, which gives us a contradiction to the linearity of $\cK^{t+1}$.
\end{proof}
%
%
Now we will prove Proposition~\ref{prop:Psi-MultiP-OneP}, which is, basically, a (measurable) version of Theorem~\ref{thm:1-Per} stating the existence of a (local) maximizer in the one-period model at time $t$.
\begin{proof}[Proof of Proposition~\ref{prop:Psi-MultiP-OneP}]
Recall that, by Remark~\ref{rem:on assumption 2.1}, the map $\widetilde\Psi_{t}(\omega^t,\cdot)$ is concave, proper with $h^{\circ,t} \in \R^{dt}$ in the interior of its domain for each $\omega^t$. Hence by \cite[Theorem~2.35, p.59]{RockafellarWets}, we have for each $\omega^t$ that
\begin{equation*}
\Psi_t(\omega^t,x^t) = \lim_{\lambda\nearrow 1}\widetilde \Psi_t(\omega^t, \lambda x^t + (1-\lambda)h^{\circ,t}).
\end{equation*}
To prove \emph{(i)}, we want to show that for $\fP$-quasi every $\omega^t$ the mapping $x^{t+1}\mapsto\Phi_t(\omega^t,x^{t+1})$ satisfies the conditions of Proposition~\ref{prop:thm9.2R} with the linear mapping $A$ being just the restriction $A(x^t,x_t)=x^t$. Fix an $\omega^t\in\Omega^t$. We deduce from the identity in \eqref{eq:psi-phi-infty-multiP} that 
$\Phi_t^\infty(\omega^t,0,\dots,0,x_t)\geq 0$ if and only if $\Psi^\infty_{t+1}(\omega^t\otimes_t\,\cdot,0,x_t)\geq 0$ $\ \fP_t(\omega^t)$-q.s., i.e. if $x_t\in K_t(\omega^t)$. We know from Proposition~\ref{prop:local-NA} that $\mbox{NA}_t$ holds, which means that $K_t(\omega^t)$ is linear for $\fP$-quasi every $\omega^t$. Thus, for $\fP$-quasi every $\omega^t$, the conditions of   Proposition~\ref{prop:thm9.2R} are indeed satisfied and hence $\widetilde\Psi_t(\omega^t,\cdot)$ is $\fP$-q.s. an upper-semicontinuous function. Moreover, from the definition of $\Psi_t(\omega^t,\cdot)$ being the upper-semicontinuous hull of $\widetilde\Psi_t(\omega^t,\cdot)$, statement \emph{(i)} follows.

We now prove \emph{(ii)}.
By Lemma~\ref{le:ass-psi-phi-multiP} and Remark~\ref{rem:ass-normal-multiP} we know that $\Phi_t$ is an $\cF_t$-normal integrand. 
Having chosen a strategy $H^t \in \cH^t$, \cite[Proposition~14.45, p.669]{RockafellarWets} yields that the mapping $\Phi^{H^t}(\omega^t,x):=\Phi_t(\omega^t,H^t(\omega^t),x)$ is a $\cF_t$-normal integrand, too. Therefore, we deduce from \cite[Theorem~14.37, p.664]{RockafellarWets} that the set-valued mapping $\Upsilon\colon\Omega^t \rightrightarrows \R^d$ defined by
\begin{equation*}
\Upsilon(\omega^t):=\mathrm{argmax}\ \Phi^{H^t}(\omega^t,\cdot)
\end{equation*}
admits an $\cF_t$-measurable selector $\widehat{h}_t$ on the universally measurable set $\{\Upsilon\neq \emptyset\}$. Extend $\widehat{h}_t$ by setting $\widehat{h}_t=0$ on $\{\Upsilon=\emptyset\}$. As we know from Proposition~\ref{prop:local-NA} that 
$\mbox{NA}_t$ holds, the attainment of the supremum in Proposition~\ref{prop:thm9.2R} gives that $\{\Upsilon=\emptyset\}$ is a $\fP$-polar set. Thus, the result follows, as $\Psi_t=\widetilde \Psi_t \ \fP$-q.s.
\end{proof}
%
Next, to see that $\Psi_{t+1}$ satisfying Assumption~\ref{ass:Psi-MultiP} implies that  $\Psi_{t}$ does, too, it remains to show that $\Psi_t$ is lower semianalytic. To that end, we first need the following useful lemma.
\begin{lemma}\label{le:dense-MultiP}
Let $g\colon \R^n\times\R^m\rightarrow\R\cup\{-\infty\}$ be a concave upper-semicontinuous function having $(x^\circ,y^\circ) \in \R^n \times \R^m$  in the interior of its domain. Then, the function 
$$
h(x) := \lim_{\lambda\nearrow 1}\sup_{y\in \R^m} g(\lambda x+(1-\lambda) x^{\circ},y)
$$
is upper-semicontinuous and satisfies
$$
h(x)  =\lim_{\lambda\nearrow 1}\sup_{y\in \Q^m} g(\lambda x+(1-\lambda) x^{\circ},y).
$$
\end{lemma}
\begin{proof}
Define the function
\begin{equation*}
\widetilde{h}(x):= \sup_{y\in \R^m} g(x,y).
\end{equation*}
Then, by \cite[Theorem~2.35, p.59]{RockafellarWets}, 
$h(x)=\mbox{cl}\,\widetilde{h}(x)$, i.e. $h$ is the upper-semicontinuous hull of $\widetilde{h}$; in particular, it is upper-semicontinuous.
Now, denote by 
\begin{align*}
\dom g&:=\{(x,y)\in\R^n\times\R^m\,|\, g(x,y)>-\infty\},\\
\dom \widetilde{h}&:=\{x\in\R^n\,|\, \widetilde{h}(x)>-\infty\}
\end{align*}
the domains of the functions $g$ and $\widetilde{h}$, respectively. 
We have to differentiate two cases.

\textsc{Case 1:} Let $x\in\overline{\dom \widetilde{h}}$. Then for each $\lambda\in(0,1)$ we have  $\widetilde{h}(\lambda x+(1-\lambda) x^{\circ})>-\infty$ and also 
$$
D_{\lambda x+(1-\lambda) x^{\circ}}:=\dom g\cap\{\lambda x+(1-\lambda) x^{\circ}\}\times\R^m\not=\varnothing.
$$
Denote by $\Pi\colon\R^n\times\R^m\rightarrow\R^m$ the projection on the second component. The set $\Pi D_x$ does not necessarily have a nonempty interior, but $\Pi D_{\lambda x+(1-\lambda) x^{\circ}}$ has, by assumption that an open ball around $(x^{\circ},y^\circ)$ is included in the domain of $g$. Hence using \cite[Theorem~2.35, p.59]{RockafellarWets} yields
$$
h(\lambda x+(1-\lambda) x^{\circ})
= \widetilde{h}(\lambda x+(1-\lambda) x^{\circ})
=\sup_{y\in \Q^m} g(\lambda x+(1-\lambda) x^{\circ},y).
$$
Taking the limit as $\lambda\nearrow 1$, using 
the upper-semicontinuity of $h$ proves the claim in the first case. 

\textsc{Case 2:} Let $x\not\in\overline{\dom \widetilde{h}}$. In this case, there exists a $\lambda_m\in(0,1)$, such that $\widetilde{h}(\lambda x+(1-\lambda) x^{\circ})=-\infty$ for all $\lambda>\lambda_m$. This implies that the set $D_{\lambda x+(1-\lambda) x^{\circ}}$ defined above is empty for each $\lambda>\lambda_m$, which yields the claim.
\end{proof}
%
%
\begin{lemma}\label{le:Psi-t-lsa-MultiP}
Fix $t \in \{0,\dots,T-1\}$. If $\Psi_{t+1}$ satisfies Assumption~2.1, then the map $\Psi_t$ is lower semianalytic.
\end{lemma}
\begin{proof}
By Lemma~\ref{le:ass-psi-phi-multiP}, the map
\begin{equation*}
\Omega^t\times \R^{d t} \times \R^d \to \overline{\R}, \quad (\omega^t,x^t,\tilde{x})\mapsto 
\Phi_t(\omega^t,x^t,\tilde{x})
\end{equation*} 
is lower semianalytic. Lemma~\ref{le:dense-MultiP} now yields
\begin{equation*}
\Psi_t(\omega^t,x^t)=\lim\limits_{\lambda\nearrow 1} \sup_{\tilde{x} \in \Q^d} \Phi_t(\omega^t, \lambda x^t + (1-\lambda) h^{\circ,t},\tilde{x}).
\end{equation*} 
This implies that $\Psi_t$ is lower semianalytic due to the fact that countable supremum of lower semianalytic functions is again lower semianalytic and a limit of a sequence of lower semianalytic functions is again lower semianalytic, see   Lemma~\cite[Lemma~7.30(2), p.178]{BertsekasShreve.78}.
\end{proof}
%
%
%
%
\begin{proof}[Proof of Proposition~\ref{prop:no-arbitrage-up-time-t-MultiP}]
The structure of the proof is similar to the one of  Proposition~\ref{prop:local-NA}. Assume by contradiction that $\cK^{t}$ is not linear. Then, there exists a probability measure in $\fP$ with its restriction to  $\Omega^t$ denoted by $P^t$, and $\widetilde{H}^t\in \cH^t$ such that
\begin{equation*}
\Psi^{\infty}_t(\widetilde H^t)\geq 0 \quad\fP\mbox{-q.s.} \quad \ \mbox{ and } \ \quad P^t[\Psi^{\infty}_t(-\widetilde H^t)<0]>0.
\end{equation*}
\textsc{Step 1:} We claim that there exists an $\cF_t$-measurable map $\widetilde{h}_t:\Omega^t \to \R^d$ such that
\begin{equation*}
\Phi^\infty_{t}(\widetilde{H}^t,\widetilde{h}_t)= \Psi^\infty_{t}(\widetilde{H}^t)\quad \fP\mbox{-q.s.}
\end{equation*}
Indeed, applying Proposition~\ref{prop:thm9.2R} to the function $f(\cdot)=\Phi_t^\infty(\omega^t,\cdot)$ yields that the set-valued map
\begin{equation*}
M(\omega^t):=\big\{h \in \R^d \,\big| \,\Phi_{t}^\infty(\omega^t, \widetilde H^t,h)=\Psi_t^\infty(\omega^t, \widetilde H^t)\big\}
\end{equation*}
is not empty for $\fP$-quasi every $\omega^t$. 
Then, following the proof of Lemma \ref{prop:Psi-MultiP-OneP}(ii) using \cite[Theorem~14.37, p.664]{RockafellarWets} provides existence of an $\cF_t$-measurable selector $\widetilde{h}_t$ of $M$.

\textsc{Step 2:} Let us show that $\widetilde{H}^{t+1}:=(\widetilde H^t,\widetilde h_t) \in \cH^{t+1}$ satisfies $\Psi^{\infty}_{t+1}(\widetilde H^{t+1})\geq 0$ $\fP\mbox{-q.s.}$, i.e. $\widetilde{H}^{t+1} \in \cK^{t+1}$.

For $\fP$-q.e. $\omega^t$ we have that
\begin{align*}
0= \Psi^{\infty}_t(\omega^t,\widetilde H^t(\omega^t)) =\Phi^{\infty}_{t}(\omega^t,\widetilde H^{t+1}(\omega^t))=\inf_{P \in \fP_t(\omega^t)} E^P[\Psi^\infty_{t+1}(\omega^t\otimes_t \cdot,\widetilde H^{t+1}(\omega^t))].
\end{align*}
As every $P'\in\fP$ satisfies $P'|_{\Omega^{t+1}}=P'|_{\Omega^{t}}\otimes P'_t$ for some selection $P'_t\in\fP_t$, we obtain the result directly from Fubini's theorem.

\textsc{Step 3:} We want to show that $-\widetilde{H}^{t+1} \notin \cK^{t+1}$. To see this,  recall the probability measure $P^t$ on $\Omega^t$ introduced at the beginning of the proof. We first modify  $\widetilde{H}^{t+1}$ on a $P^t$-nullset to obtain a Borel measurable function $\widetilde{H}^{P^t,t+1}$. Consider the set
\begin{equation*}
\fS_{t}:=\Big\{(\omega^{t},P) \in \Omega^{t}\times \fM_1(\Omega_1)\,\Big| \,P \in \fP_{t}(\omega^t), \ E^P\big[\Psi^\infty_{t+1}(\omega^t\otimes_t\cdot,-\widetilde H^{P^t, t+1}(\omega^t))\big]<0 \Big\}.
\end{equation*}
Using the same arguments as in Step~2 and  Step~3 of Proposition~\ref{prop:local-NA} we get that $\fS_{t}$ is analytic, hence there exists an universally measurable kernel $P_t:\Omega^t\mapsto \fM_1(\Omega_1)$ such that $(\omega^t,P_t(\omega^t))\in \fS_t$ for all $\omega^t \in \mbox{proj }\fS_t$. 

We claim that $P^t[\proj\fS_t]>0$. To see this, observe first that $\fP$-q.s., we have for any $x^t \in \R^{dt}$, $x' \in \R^d$ that $\Psi_t(x^t)\geq \Phi_t(x^t,x')$. Hence, we obtain from \cite[Theorem~3.21, p.88]{RockafellarWets} that $\fP$-q.s., we also have $\Psi^\infty_t(x^t)\geq \Phi^\infty_t(x^t,x')$ for any $x^t \in \R^{dt}$, $x' \in \R^d$. Therefore, we have for $\fP$-quasi every $\omega^t$ that
\begin{align*}
\Psi_t^\infty(\omega^t,-\widetilde H^{t}(\omega^t))\geq \Phi_t^\infty(\omega^t,-\widetilde H^{t+1}(\omega^t))
=\inf_{P \in \fP_t(\omega^t)}E^P\big[\Psi_{t+1}^\infty(\omega^t\otimes_t \cdot,-\widetilde H^{t+1}(\omega^t))\big].
\end{align*}
As $P^t[\Psi_t^\infty(-\widetilde H^{t})]<0]>0$ and $H^{t+1}=H^{P^t, t+1} \ P^t$-a.s., we conclude that indeed $P^t[\mbox{proj }\fS_t]>0$. Finally, to see that $-H^{t+1}\not\in\cK^{t+1}$,
define the measure $\widetilde{P} \in \fP$ by
\begin{equation*}
\widetilde{P}:=P^t|_{\Omega^t}\otimes P_t\otimes \widetilde P_{t+1}\otimes\dots\otimes \widetilde  P_{T-1}, 
\end{equation*}
where we take any selector $\widetilde P_{s} \in \fP_s$ for $s:=t+1,\dots T-1$. Then, by construction, we have
\begin{equation*}
E^{\widetilde P}[\Psi^\infty_{t+1}(-\widetilde H^{t+1})]<0,
\end{equation*}
hence indeed $-H^{t+1}\not\in\cK^{t+1}$, which gives a contradiction to the linearity of $\cK^{t+1}$.
\end{proof}
%
%
%
\begin{proof}[Proof of Proposition~\ref{prop:induction}]
We have shown that if $\Psi_{t+1}$ satisfies Assumption~\ref{ass:Psi-MultiP} and $\cK^{t+1}$ is linear (i.e. $\mbox{NA}(\fP)^{t+1}$ holds), then $\Psi_{t}$ satisfies Assumption~\ref{ass:Psi-MultiP} and the local no-arbitrage condition $\mathrm{NA}_t$ holds. 
Moreover, we have shown that the linearity of $\cK^{t+1}$ implies the linearity of $\cK^{t}$ as soon as $\Psi_{t+1}$ satisfies Assumption~\ref{ass:Psi-MultiP}. As
by assumption, $\Psi=\Psi_{T}$ satisfies Assumption~\ref{ass:Psi-MultiP} and $\mbox{NA}(\fP)$ holds, 
we see that Proposition~\ref{prop:induction} holds by using backward induction.
\end{proof}
%
%
\subsection{Proof of Theorem~\ref{thm:Maxim-exist-MultiP}}\label{subsec:Proof-main-thm}
The goal of this subsection is to give the proof of Theorem~\ref{thm:Maxim-exist-MultiP}, which is the main result of this paper. We will construct the optimal strategy $\widehat{H}:=(\widehat H_0,\dots,\widehat H_{T-1}) \in \cH$ recursively from time $t=0$ upwards by applying Proposition~\ref{prop:Psi-MultiP-OneP} at each time $t$, given the  restricted strategy $\widehat{H}^t:=(\widehat H_0,\dots,\widehat H_{t-1})\in \cH^t$. We follow \cite{Nutz.13util} to check that $\widehat{H}$ is indeed an optimizer of \eqref{eq:thm-optimal-MultiP}.
%
%
%
\begin{lemma}\label{le:Prob-Meas-Approx-MultiP}
Let $t \in \{0,\dots,T-1\}$ and $H^{t+1}\in\cH^{t+1}$. Define the random variable 
\begin{equation*}
X(\omega^t):=\Phi_t(\omega^t,H^{t+1}(\omega^t)).
\end{equation*}
For any $\varepsilon>0$, there exists a universally measurable kernel $P^\varepsilon_t:\Omega^t \to \fM_1(\Omega_1)$ such that $P^\varepsilon_t(\omega^t) \in \fP_t(\omega^t)$ for all $\omega^t \in \Omega^t$ and 
\begin{equation*}
E^{P^\varepsilon_t(\omega^t)}[\Psi_{t+1}(\omega^t\otimes_t\cdot, H^{t+1}(\omega^t))]\leq
\begin{cases}
X(\omega^t) + \eps & {\rm{if }}\ X(\omega^t)>-\infty, \\
-\eps^{-1} & {\rm otherwise.}
\end{cases}
\end{equation*}
\end{lemma}
\begin{proof}
For any $x^{t+1}\in \R^{d(t+1)}$, define the random variable
\begin{equation*}
\Phi^{x^{t+1}}(\omega^t):=\Phi_t(\omega^t,x^{t+1}).
\end{equation*}
By the proof of Lemma~\ref{le:ass-psi-phi-multiP},
the map $(\omega^t,P,x^{t+1})\mapsto E^P[\Psi_{t+1}(\omega^t\otimes_t\,\cdot,x^{t+1})]$ is lower semianalytic. Moreover, by assumption,  $\mbox{graph}(\fP_t)$ is analytic. Therefore, by \cite[Theorem~7.50,p.184]{BertsekasShreve.78}, it admits a universally measurable kernel $(\omega^t,x^{t+1})\mapsto \tilde{P}^{\varepsilon}_t(\omega^t,x^{t+1}) \in \fP_t(\omega^t)$ satisfying 
\begin{equation*}
E^{\tilde{P}^\varepsilon_t(\omega^t,x^{t+1})}[\Psi_{t+1}(\omega^t\otimes_t \cdot, x^{t+1})]\leq
\begin{cases}
\Phi^{x^{t+1}}(\omega^t) + \eps & \text{if }
\Phi^{x^{t+1}}_t(\omega^t)>-\infty, \\
-\eps^{-1} & {\rm otherwise.}
\end{cases}
\end{equation*}
Setting $P^\varepsilon_t(\omega^t):=\tilde{P}^\varepsilon(\omega^t,H^{t+1}(\omega^t))$ yields the result, as $X(\omega^t)=\Phi^{H^{t+1}( \omega^t)}(\omega^t)$.
\end{proof}
%
%
Now we are able to prove Theorem~\ref{thm:Maxim-exist-MultiP}.
\begin{proof}[Proof of Theorem~\ref{thm:Maxim-exist-MultiP}]
By Theorem~\ref{thm:1-Per},
there exists $\widehat{H}_0\in \R^d$ such that
\begin{equation*}
\inf_{P_\in \fP_0} E^P[\Psi_1(\widehat{H}_0)]=\sup_{x\in \R^d} \inf_{P_\in \fP_0} E^P[\Psi_1(x)].
\end{equation*} 
By a recursive application of Proposition~\ref{prop:Psi-MultiP-OneP}, we can define an $\cF_t$-measurable random variable $\widehat{H}_t$ such that
\begin{equation*}
\inf_{P \in \fP_t(\omega^t)}E^P[\Psi_{t+1}(\omega^t\otimes_t \cdot,\widehat H^t(\omega^{t-1}),\widehat{H}_{t}(\omega^t))]
=\Psi_t(\omega^t,\widehat H^t(\omega^{t-1}))
\end{equation*}
for $\fP$-quasi-every $\omega^t\in \Omega^t$, for all $t=1,\dots T-1$. We claim that $\widehat H\in\cH$ is optimal, i.e. satisfies \eqref{eq:thm-optimal-MultiP}.
We first show that
\begin{equation}\label{eq:pf-thm-claim1-MultiP}
\inf_{P\in \fP}E^P[\Psi_T(\widehat H)]\geq \Psi_0.
\end{equation}
To that end, let $t \in \{0,\dots,T-1\}$. Let $P \in \fP$; we write $P= P_0\otimes\dots\otimes P_{T-1}$ with kernels $\fP_s:\Omega^s\to\fM_1(\Omega_1)$  satisfying $P_s(\cdot)\in \fP_s(\cdot)$. Therefore, by applying Fubini's theorem and the definition of $\widehat H$
\begin{align*}
\ E^P[&\Psi_{t+1}(\widehat{H}_0,\dots,\widehat{H}_{t})]\\[1ex]
&=  \ E^{(P_0\otimes\dots\otimes P_{t-1})(d\omega^t)}\Big[E^{P_t(\omega^t)}\big[\Psi_{t+1}(\omega^t\otimes_t\,\cdot,\widehat{H}^t(\omega^{t-1}),\widehat{H}_{t}(\omega^t))\big]\Big]\\[1ex]
&\geq  \ E^{(P_0\otimes\dots\otimes P_{t-1})(d\omega^t)}\Big[\inf_{P' \in \fP_t(\omega^t)}E^{P'}\big[\Psi_{t+1}(\omega^t\otimes_t\,\cdot,\widehat H^t(\omega^{t-1}),\widehat{H}_{t}(\omega^t))\big]\Big]\\[1ex]
&= \ E^{(P_0\otimes\dots\otimes P_{t-1})}[\Psi_t(\widehat H^t)]\\[1ex]
&= E^P[\Psi_t(\widehat H^t)].
\end{align*}
Using this inequality repeatedly from $t=T-1$ to $t=0$ yields $E^P[\Psi_T(\widehat H)]\geq \Psi_0$. As $P \in \fP$ was arbitrarily chosen, the claim \eqref{eq:pf-thm-claim1-MultiP} is proven. It remains to show that
\begin{equation*}
\Psi_0 \geq \sup_{H \in \cH}\inf_{P \in \fP}E^P[\Psi(H)] 
\end{equation*}
to see that $\widehat{H} \in \cH$ is optimal. So, fix an arbitrary $H \in \cH$. It suffices to show that for every $t \in \{0,\dots,T-1\}$
\begin{equation}\label{eq:pf-thm-claim2-MultiP}
\inf_{P \in \fP}E^P[\Psi_t(H^t)]
\geq
\inf_{P \in \fP}E^P[\Psi_{t+1}(H^{t+1})].
\end{equation}
Indeed, using the inequality repeatedly from $t=0$ until $t=T-1$ yields
\begin{equation*}
\Psi_0\geq \inf_{P\in \fP}E^{P}[\Psi_{T}( H)].
\end{equation*}
Furthermore, as $H \in \cH$ was arbitrary and $\Psi_T=\Psi$, we obtain the desired inequality
\begin{equation*}
\Psi_0\geq \sup_{H \in \cH}\inf_{P\in \fP}E^{P}[\Psi( H)].
\end{equation*}

Now, to prove the inequality in~\eqref{eq:pf-thm-claim2-MultiP}, fix an $\varepsilon>0$. By Lemma~\ref{le:Prob-Meas-Approx-MultiP}, there exists a kernel $P^\varepsilon_t:\Omega^t\to\fM_1(\Omega_1)$ such that for all $\omega^t \in \Omega^t$
\begin{align*}
\ E^{P^\varepsilon_t(\omega^t)}&[\Psi_{t+1}(\omega^t\otimes_t\,\cdot, H^{t+1}(\omega^t))]-\varepsilon\\[1ex]
&\leq  \ (-\varepsilon^{-1})\vee \inf_{P \in \fP_t(\omega^t)}E^P[\Psi_{t+1}(\omega^t\otimes_t\,\cdot, H^{t+1}(\omega^t))]\\[1ex]
&\leq  \ (-\varepsilon^{-1})\vee \sup_{x \in \R^d}\inf_{P \in \fP_t(\omega^t)}E^P[\Psi_{t+1}(\omega^t\otimes_t\,\cdot, H^t(\omega^{t-1}),x)]\\[1ex]
& = \ (-\varepsilon^{-1})\vee \Psi_t(\omega^t,H^t(\omega^{t-1})).
\end{align*}
Take any $P \in \fP$ and denote its restriction to $\Omega^t$ by $P^t$. Integrating the above inequalities yields
\begin{align*}
E^{P^t}[(-\varepsilon^{-1})\vee \Psi_t(H^t)]
\geq
E^{P^t\otimes P^\varepsilon_t}[\Psi_{t+1}( H^{t+1})]-\varepsilon 
\geq
\inf_{P'\in \fP}E^{P'}[\Psi_{t+1}( H^{t+1})]-\varepsilon.
\end{align*}
Letting $\varepsilon \to 0$, we obtain, by Fatou's Lemma, that
\begin{equation*}
E^P[ \Psi_t(H^t)] \geq \inf_{P'\in \fP}E^{P'}[\Psi_{t+1}( H^{t+1})].
\end{equation*}
This implies the inequality~\eqref{eq:pf-thm-claim2-MultiP}, as $P \in \fP$ was arbitrary. 
\end{proof}
%
%
\appendix
\section*{Appendix}\label{sec:Appendix}
Here we provide a simple fact about horizon functions of compositions. Let us first recall the definition of the domain of a function $f\colon\R^n\rightarrow\R$
$$
\dom f :=\{ x \in \R^n \,|\,f(x)>-\infty\}.
$$
%
%
\begin{lemma}\label{le:comp-horizon-append}
Let $U\colon \R \to \R \cup\{-\infty\}$ be concave, nondecreasing, nonconstant and upper-semicontinuous. Let $V\colon \R^n \to \R$ be concave, upper-semicontinuous and assume that $V(\R^n)\cap \dom U \neq \emptyset$. Then the function 
\begin{equation*}
\Psi:\R^n \to \R\cup\{-\infty\}, \quad  h \mapsto  \Psi(h):=
\begin{cases}
U(V(h)) & \mbox{if } \  h \in \dom V\\
-\infty& \mbox{otherwise}.
\end{cases}
\end{equation*}
is concave, proper and upper-semicontinuous. Moreover,   $\Psi^\infty$ satisfies
\begin{equation*}
\Psi^\infty(h)=\begin{cases}
U^\infty(V^\infty(h)) & \mbox{if } \  h \in \dom V^\infty\\
-\infty& \mbox{otherwise}.
\end{cases}
\end{equation*}
\end{lemma}
\begin{proof}
The first part of the lemma is obvious. The only thing requiring a proof is the statement about the form of the horizon function $\Psi^\infty$.

First, choose a point $x\in \R^n$ such that $V(x)\in\dom U$; equivalently, such that $\Psi(x)>-\infty$. Now, fix any $h \in \R^n$. The mapping $V$ is concave, hence 
$$
\mbox{the sequence\ }\ 
m\mapsto\frac1m \big(V(x+mh)-V(x)\big)=:a_m
\ \ 
\mbox{is nonincreasing.}
$$
Denote its limit by $a$. Then, $a>-\infty$ if and only if $h\in\dom V^\infty$; indeed $a=V^\infty(h)$.

Let us first estimate the horizon function $\Psi^\infty$ from above. 
If $a_m=-\infty$ for some $m$, this implies by definition that $\Psi^\infty(h)=-\infty$, as $(a_m)$ is nonincreasing. Now assume that $a_m>-\infty$ for each $m$. Then 
\begin{align*}
\Psi^\infty(h)
&= 
\lim_{n\rightarrow\infty}\frac1n\big(\Psi(x+nh)-\Psi(x)\big) \\
&=
\lim_{n\rightarrow\infty}\frac1n \left(U\left(n\,\frac1n \big(V(x+nh)-V(x)\big) + V(x)\right)-\Psi(x)\right)\\
&\leq
\lim_{n\rightarrow\infty}\frac1n \big(U\left(n\,a_m + V(x)\right)-U(V(x))\big)\\
&=
U^\infty(a_m).
\end{align*}
If $h \in \dom V^\infty$, then by letting $m$ tend to infinity, the above estimate and upper-semicontinuity of $U^\infty$ yield $\Psi^\infty(h)\leq U^\infty(V^\infty(h))$. 
If $h \notin \dom V^\infty$, then $a_m$ diverges to $-\infty$, hence also $U(a_m)$ tends to $-\infty$ as $U$ is concave, nondecreasing and nonconstant.
By definition of the horizon function we have
$$
U^\infty(a_m)\leq U(a_m+V(x))-U(V(x)),
$$ 
hence also $U^\infty(a_m)$ tends to $-\infty$. 
This proves the desired first inequality.

To estimate  $\Psi^\infty$ from below, we only need to consider $h \in \dom V^\infty$. Indeed, we know from above that  $\Psi^\infty(h)=-\infty$ whenever $h \notin \dom V^\infty$. So let $h \in \dom V^\infty$. Then
\begin{align*}
\Psi^\infty(h)
=
\lim_{n\rightarrow\infty}\frac1n \big(U(n\,a_n+V(x))-U(V(x))\big)
&\geq
\lim_{n\rightarrow\infty}\frac1n \big(U(n\,a + V(x))-U(V(x))\big)\\
&= U^\infty(V^\infty(h)).
\end{align*}
\end{proof}

%
%
	
%
%
%

%
%




\end{document}